\documentclass[journal]{IEEEtran}
\usepackage{cite}
\usepackage{amsmath}
\usepackage{amssymb}
\usepackage{floatfig,epsfig}
\usepackage{subfigure}
\usepackage{verbatim}
\usepackage{pifont}
\usepackage{algorithm}
\usepackage{algorithmic}
\usepackage{alltt}
\usepackage{bm}
\usepackage{isolatin1}
\usepackage{graphics}
\usepackage{graphicx,dblfloatfix}
\usepackage{epic}
\usepackage{eepic}
\usepackage{graphpap}
\usepackage{psfrag}
\usepackage{threeparttable}
\usepackage{arydshln}

\usepackage[usenames,dvipsnames]{pstricks}
\usepackage{epsfig}
\usepackage{pst-grad} % For gradients
\usepackage{pst-plot} % For axes

\newtheorem{definition}{Definition}
\newtheorem{remark}{Remark}
\newtheorem{lemma}{Lemma}
\newtheorem{theorem}{\textbf{Theorem}}

\normalsize

\ifCLASSINFOpdf
\else
\fi

\begin{document}


\title{Secure Degrees of Freedom of Wireless X Networks Using Artificial Noise Alignment}

\author{Zhao~Wang, \IEEEmembership{Student Member, IEEE,}
        Ming~Xiao, \IEEEmembership{Senior Member, IEEE,}
        Mikael~Skoglund, \IEEEmembership{Senior Member, IEEE,}
        and H. Vincent Poor, \IEEEmembership{Fellow, IEEE}% <-this % stops a space
        \thanks{Z. Wang, M. Xiao and M. Skoglund are with the Department of Communication Theory, School of Electrical Engineering, Royal Institute of Technology (KTH), Stockholm,
Sweden (E-mail:\{zhaowang, mingx, skoglund\}@kth.se).

         H. V. Poor is with the Department of Electrical Engineering, Princeton University, Princeton, NJ (E-mail:poor@princeton.edu).
         
         The research was supported in part by the U.S. National Science Foundation under Grant CMMI-1435778.}
         }

\markboth{Submitted for publication}
{Zhao Wang}

\vspace{-6ex}

\maketitle

\begin{abstract}
The problem of transmitting confidential messages in $M \times K$ wireless X networks is considered, in which each transmitter intends to send one confidential message to every receiver. In particular, the secure degrees of freedom (SDOF) of the considered network are studied based on an artificial noise alignment (ANA) approach, which integrates interference alignment and artificial noise transmission. At first, an SDOF upper bound is derived for the $M \times K$ X network with confidential messages (XNCM) to be $\frac{K(M-1)}{K+M-2}$. By proposing an ANA approach, it is shown that the SDOF upper bound is tight when $K=2$ for the considered XNCM with time/frequency varying channels. For $K \geq 3$, it is shown that SDOF of $\frac{K(M-1)}{K+M-1}$ can be achieved, even when an external eavesdropper is present. The key idea of the proposed scheme is to inject artificial noise into the network, which can be aligned in the interference space at receivers for confidentiality. Moreover, for the network with no channel state information at transmitters, a blind ANA scheme is proposed to achieve SDOF of $\frac{K(M-1)}{K+M-1}$ for $K,M \geq 2$, with reconfigurable antennas at receivers. The proposed method provides a linear approach to secrecy coding and interference alignment.
\end{abstract}

\begin{IEEEkeywords}
	Secure degrees of freedom, artificial noise, interference alignment, wireless X network
\end{IEEEkeywords}

\section{Introduction}

\subsection{Background and Motivation}

The notion of secrecy capacity was introduced by Wyner \cite{Wyner1975} in the context of the wire-tap channel, in which a legitimate transmitter intends to send a confidential message to a legitimate receiver by hiding it from a degraded eavesdropper. Later the non-degraded wire-tap channel \cite{Csiszar1978} and Gaussian wire-tap channel \cite{Cheong1978} were studied to generalize Wyner's work. In recent years, multiuser secret communications has drawn substantial research attention. For example, the interference channel and broadcast channel with secrecy constraints were studied in \cite{RLiu2008} and \cite{JXu2009}, multiple access channels with secrecy constraints were investigated in \cite{Tekin2008}, \cite{Liang2008} and \cite{Awan2013}, the relay-eavesdropper channel was studied in \cite{Lai2008}, and parallel relay channels were considered in \cite{Awan2012} and \cite{ZLi2006}. Usually the exact secrecy capacity region is difficult to find for most multiuser networks. As a consequence, the secure degrees of freedom (SDOF) which serves as an approximation of the secrecy capacity in the high signal-to-noise ratio (SNR) regime has been widely investigated recently \cite{OOKoy2011,GBag2010,XHe2010,AKhisti2013,Jianwei2012,Jianwei2013}.

The secrecy capacity of the original wire-tap channel \cite{Wyner1975} is essentially the mutual information difference between the legitimate user pair and transmitter-eavesdropper pair, which renders a vanishing SDOF in the high SNR regime. However, positive SDOF can be achieved for some other multiuser networks, e.g., the multi-antenna compound wiretap channel \cite{Khisti2011}, the interference channel \cite{OOKoy2011,Jianwei2013}, the broadcast and multiple access channels with confidential messages \cite{Jianwei2012} \emph{et al}. These results reveal the fact that interference can have positive impact on the secrecy capacity of networks, because it naturally serves as jamming to conceal messages from eavesdroppers. The assistance of interference in secure communication is well addressed in \cite{Xiaojun2011}.

As a novel approach to handle interference in multiuser networks, interference alignment (IA)\cite{Cadambe2008} provides advantages for limiting the information leakage of confidential messages. Intuitively speaking, IA can pack the undesired messages into a reduced dimensional interference subspace at receivers, where the signals containing confidential messages are superimposed. Therefore, it naturally brings difficulty for the receiver when it tries to decode the information from the interference subspace. It is first noted in \cite{OOKoy2011} that by combining IA and random binning \cite{Wyner1975}, SDOF of $\frac{K(K-1)}{2(K-2)}$ can be achieved in the time/frequency varying $K$-user Gaussian interference channel. By adopting the Wyner random binning method to provide a secret codebook, IA has been generalized to different networks to obtain positive SDOF, e.g., the multi-antenna compound wiretap channel \cite{Khisti2011}, and the multi-antenna wiretap channel with block fading channels \cite{MKobayashi2011}. The key idea of random binning is to provide randomness to the codebook, such that the eavesdropper is not able to tell the exact codeword from the randomnized codebook. From a different transmission approach, artificial noise (AN) works in another efficient way of providing randomness to the codebook, which aims to mask the confidential signal at the eavesdropper \cite{Goel2008}. The AN can be chosen to be a Gaussian process. When the power of the AN is high enough to be comparable with the message power, it can provide enough randomness with the maximum differential entropy to confuse decoding. As studied in \cite{Jianwei2012,Jianwei2013,SYang2013} and \cite{Zhao2014a}, the transmission of AN and IA can be integrated to achieve the optimal SDOF of different multiuser networks. The proposed artificial noise alignment (ANA) schemes provide a different perspective and an efficient approach for investigating the SDOF of networks: instead of random binning, we can inject AN into the confidential message subspace at eavesdroppers. By only considering signal dimensions, the approach of aligning AN to a certain subspace can offer a better transmission design than random binning in terms of SDOF. For instance, the optimal SDOF of the $K$-user Gaussian interference channel with confidential messages (ICCM) is shown to be $\frac{K(K-1)}{2K-1}$ \cite{Jianwei2013} for constant channel state, achieved by an ANA scheme. However, by random binning, only the inferior SDOF $\frac{K(K-2)}{2(K-1)}$ can be achieved \cite{OOKoy2011}. Likewise, the ANA scheme also achieves the optimal SDOF for MIMO broadcast channels \cite{SYang2013} and two hop interference channels \cite{Zhao2014a} with delayed channel state information at transmitters (CSIT), which can be seen as a non-trivial generalization of the Maddah-Ali \& Tse scheme \cite{MATse2012} for secret communications. Moreover, compared with random binning the ANA approach offers lower system complexity via its linear operations. Therefore, the ANA approach is interesting from both theoretical and practical viewpoints.    

In this paper, we focus on one of the typical network models in information theory, namely, the wireless X network. As it has both the nature of broadcast channels and multiple access channels, the wireless X network has drawn substantial research attention for studying the efficient signaling for transmission. For instance, with perfect CSIT, the wireless X networks with different antenna settings have been studied in \cite{MaddahAli08,Cadambe2009,LYang2012} based on interference alignment schemes. With delayed CSIT, several results have been reported to generalize the Maddah-Ali \& Tse scheme to wireless X networks, e.g., \cite{Ghasemi2012,Tandon2012,Zhao2012GC}. It is also notable that with secrecy constraints, the $2 \times 2$ MIMO X channel has been studied in \cite{Zaidi2013} with output feedback and delayed CSIT, where an ANA scheme is proposed based on feedback. The contributions of our work are highlighted in the following section.   

\subsection{Contribution}

We study the SDOF of the wireless X network with confidential messages (XNCM). Specifically, the SDOF of the Gaussian $M \times K$ XNCM is investigated, in which each transmitter intends to send one confidential message to each receiver. The main results can be summarized as follows. 

\emph{1) A general sum SDOF upper bound for XNCM:} We bound the sum SDOF of the $M \times K$ XNCM to be less than or equal to $\frac{K(M-1)}{K+M-2}$ regardless of channel fading variations. Therefore, the proposed bound holds for time/frequency varying channels, and/or constant channels. To compare with the interference channel counterpart, we set $K=M$ to induce the upper bound $\frac{K}{2}$. Consequently, every transmitter can obtain \emph{at most} half of the resources.

\emph{2) The optimal sum SDOF of the time/frequency varying XNCM:} With perfect CSIT, the optimal sum SDOF of the $M \times K$ XNCM with time/frequency varying channels is shown to be
\begin{align} \nonumber
	d = \frac{K(M-1)}{K+M-2}, ~\textrm{if}~K=2,
\end{align}
and
\begin{align} \nonumber
	\frac{K(M-1)}{K+M-1} \leq d \leq \frac{K(M-1)}{K+M-2}, ~\textrm{if}~K \geq 3.
\end{align}
Therefore, the upper bound is tight for the networks with two receivers. The sum SDOF lower bound is achieved by an ANA approach, which combines standard interference alignment \cite{Cadambe2008,Cadambe2009} and artificial noise transmission. We note that the sum SDOF $\frac{K(M-1)}{K+M-1}$ overlaps with the results in \cite{Tiangao2008}. However, we prove it using a new approach: by proposing an ANA scheme we show that $\frac{K(M-1)}{K+M-1}$ can be achieved for the $M \times K$ XNCM with an external eavesdropper (EE), which implies the same achieved sum SDOF for the considered network without the EE. It is also notable that when $K=M$ and $K,M \geq 3$, our result overlaps with the optimal sum SDOF of the $K$-user interference channel with confidential messages \cite{Jianwei2013}. Even with an EE, treating the $X$ network as an interference channel can still obtain the best known SDOF so far \cite{Jianwei2013}. 

\emph{3) The achieved sum SDOF of the XNCM with reconfigurable antennas:} Following a similar principle, we generalize the ANA scheme into a blind approach with the help of reconfigurable antennas at the receivers. We note that the blind ANA scheme does not require any CSIT in the network. By a predefined private antenna switching pattern, we integrate the AN into a blind IA scheme to achieve the sum SDOF $\frac{K(M-1)}{K+M-1}$ for the $M \times K$ XNCM. It is worth noting that the predefined antenna switch pattern not only provides the channel coherence structure for aligning interference \cite{Tiangao2011}, it also serves as a secret key for different receivers to decode. 

The rest of the paper is organized as follows. Section \ref{section:SYSMD} introduces the system model for the $M \times K$ XNCM. We provide the proposed SDOF upper bound in Section \ref{section:SDOFConverse}. In Section \ref{section:SDOFXNCM}, we study the SDOF of the considered network with time/frequency varying channels, where the ANA scheme is proposed to achieve the sum SDOF lower bound. We generalize the ANA into a blind approach in Section \ref{section:SDOFBlind}. Conclusions are given in Section \ref{section:CONC}.

Notation: Throughout the paper, we use bold-faced uppercase letters, plain uppercase letters, and lowercase letters ($\mathbf{X}, X, x$) to represent matrices, vectors, and scalars, respectively, unless otherwise stated. $X^{n}$ represents the sequence $\{x_{1},x_{2},\dots,x_{n}\}$. We define $\mathcal{K} = \{1,2,\dots, K\}$, and $\mathcal{M} = \{1,2,\dots, M\}$, where specifically, $K,M$ are two integers. $\mathcal{K}-i$ denotes the set $\mathcal{K}$ after removing the element $i \in \mathcal{K}$. $\otimes$ denotes the \emph{Kronecker Product}. $\mathbf{A} \prec \mathbf{B}$ means that $\textrm{span}(\mathbf{A}) \subset \textrm{span}(\mathbf{B})$, where $\textrm{span}(\cdot)$ represents the column span of the matrix.

\begin{figure*}[t!]
\centerline{\epsfig{file=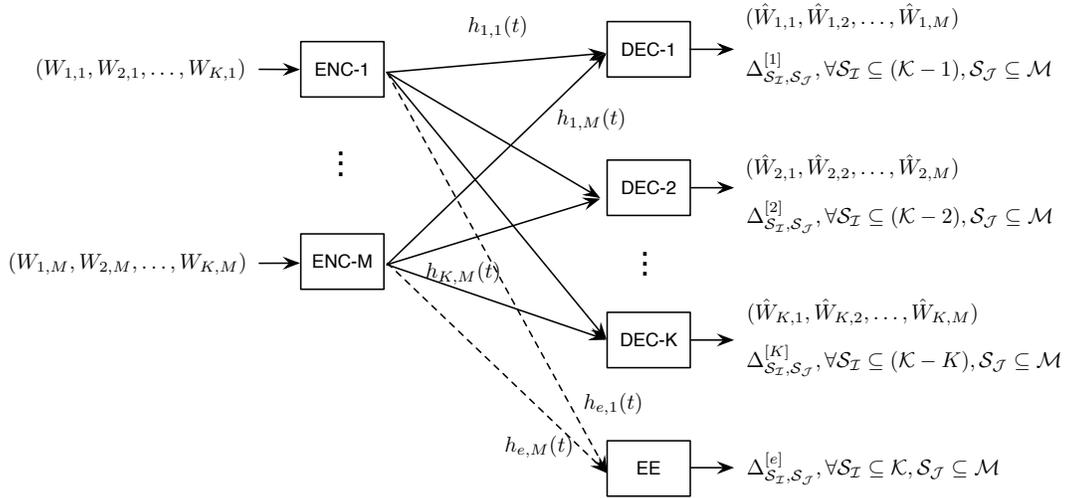,width=140mm}}
\caption{$M \times K$ XNCM with an external eavesdropper.}
\label{Fig:XNCSAS}
\end{figure*}

\section{System Model}
\label{section:SYSMD}
We mainly consider transmitting confidential messages in the wireless X network. Specifically, we provide the following definitions on the network. 

\begin{definition}\emph{$M \times K$ wireless X network with confidential messages (XNCM).}
\label{Def:XNCM}

Consider the $M \times K$ wireless X network, where each of the $M$ transmitters intends to deliver confidential messages to all $K$ receivers. Therefore, there are $MK$ confidential messages in the considered network, shown in Fig. \ref{Fig:XNCSAS}. The received signal for receiver $k$, at time $t$ is
\begin{equation} \nonumber
    y_{k}(t) = \sum_{m \in \mathcal{M}} h_{km}(t) x_{m}(t) + n_{k}(t), ~ k \in \mathcal{K},
\end{equation}
where the scalar $x_{m}(t)$ represents the transmitted signal of transmitter $m$, the scalar $h_{km}(t)$ represents the channel coefficient from transmitter $m$ to receiver $k$. To avoid channel degeneration, we assume the absolute value of channel coefficient is bounded away from zero and above by a finite value. The term $n_{k}(t) \sim \mathcal{CN}(0,1)$ is additive white Gaussian noise (AWGN). We assume the power constraint $\frac{1}{n}\sum_{t=1}^{n} |x_{m}(t)|^{2} \leq P$. We define $W_{k,m} \in \mathcal{W}_{k,m} = [1:2^{nR_{k,m}}]$, $k \in \mathcal{K}$ and $m \in \mathcal{M}$, denoting the confidential message from transmitter $m$ to receiver $k$ with the secrecy rate $R_{k,m}$. Let $\mathcal{W} = \{\mathcal{W}_{k,m}\}_{k\in \mathcal{K}, m \in \mathcal{M}}$. A secrecy rate tuple $\mathbf{R} = \{R_{k,m}\}_{k \in \mathcal{K}, m \in \mathcal{M}}$ is achievable if there exists a secret codebook $(n, \mathbf{R}, \mathcal{W})$ to satisfy the reliability and confidentiality constraints simultaneously:
\begin{itemize}
    \item {reliability: at receiver $k$,
        \begin{equation} \label{Eqn:reliability}
            \lim_{n \rightarrow \infty} \textrm{Pr}(\hat{W}_{k,m} \neq W_{k,m}) = 0, ~ \forall m \in \mathcal{M}
        \end{equation} where $\hat{W}_{k,m}$ is the estimation of the codeword $W_{k,m}$, and}
    \item {confidentiality: at receiver $k$,
    		\begin{equation}
    			\nonumber
    			\lim_{n,P\rightarrow \infty}\Delta_{\mathcal{S_{I}},\mathcal{S_{J}}}^{[k]} = 1,
    		\end{equation}
    where the equivocation for each subset of messages $\mathbf{W}_{\mathcal{S_{I}},\mathcal{S_{J}}}$ at receiver $k$ is defined as 
        \begin{equation} \label{Eqn:confidentiality}
            \Delta_{\mathcal{S_{I}},\mathcal{S_{J}}}^{[k]} \overset{\triangle}{=} \frac{H(\mathbf{W}_{\mathcal{S_{I}},\mathcal{S_{J}}}|Y_{k}^{n})}{H(\mathbf{W}_{\mathcal{S_{I}},\mathcal{S_{J}}})},
        \end{equation} where $\mathbf{W}_{\mathcal{S_{I}}, \mathcal{S_{J}}} = \{W_{ij} : i \in \mathcal{S_{I}}, j \in \mathcal{S_{J}}\}$ for all $\mathcal{S_{I}} \subseteq \mathcal{K}-k$, $\mathcal{S_{J}} \subseteq \mathcal{M}$ such that $H(\mathbf{W}_{\mathcal{S_{I}}, \mathcal{S_{J}}}) > 0$. Note that the equivocation (\ref{Eqn:confidentiality}) is only defined on the non-zero rate message set. For the zero rate message set, the confidentiality is automatically satisfied because the information leakage can only be zero.}
\end{itemize}
If for a certain power $P$, $\{R_{k,m}\}$ are achievable, the sum SDOF $d$ is said to be achieved with definition
\begin{equation} \label{Eqn:SDOFXC}
    d = \lim_{P \rightarrow \infty} \frac{\sum_{k \in \mathcal{K}, m \in \mathcal{M}}R_{k,m}}{\textrm{log}(P)}.
\end{equation}
We say $d$ is the optimal sum SDOF if it is the maximum value of all achievable sum SDOFs. We consider only the sum SDOF in this paper. 
\end{definition}

In Section \ref{section:SDOFXNCM}, we also consider the network when there exists an external eavesdropper. Inherited from the above definition on XNCM, we present that network as follows. 
\begin{definition}\emph{The XNCM with an external eavesdropper (XNCM-EE).}

Consider the $M \times K$ XNCM, when an external eavesdropper $e$ appears in the network. The received signal at the eavesdropper $e$, at time $t$ is 
\begin{equation} \nonumber
	y_{e}(t) = \sum_{m \in \mathcal{M}} h_{em}(t) x_{m}(t) + n_{e}(t), ~ k \in \mathcal{K}. 
\end{equation}
Following \emph{Definition} \ref{Def:XNCM}, we say a secrecy rate tuple $\mathbf{R}$ is achievable for a transmission power $P$ if there exists a secret codebook $(n, \mathbf{R}, \mathcal{W})$ to satisfy reliability (\ref{Eqn:reliability}), confidentiality (\ref{Eqn:confidentiality}) and also an extra secrecy constraint at the eavesdropper:
\begin{equation} \label{Eqn:secEva}
	\lim_{n,P \rightarrow\infty}\Delta_{\mathcal{S_{I}},\mathcal{S_{J}}}^{[e]} \overset{\triangle}{=} \lim_{n,P \rightarrow \infty}\frac{H(\mathbf{W}_{\mathcal{S_{I}},\mathcal{S_{J}}}|Y_{e}^{n})}{H(\mathbf{W}_{\mathcal{S_{I}},\mathcal{S_{J}}})} = 1,
\end{equation} for $\mathbf{W}_{\mathcal{S_{I}}, \mathcal{S_{J}}} = \{W_{ij} : i \in \mathcal{S_{I}}, j \in \mathcal{S_{J}}\}$ and $H(\mathbf{W}_{\mathcal{S_{I}}, \mathcal{S_{J}}}) > 0$, for all $\mathcal{S_{I}} \subseteq \mathcal{K}$, $\mathcal{S_{J}} \subseteq \mathcal{M}$. 
\end{definition}

\section{Secure Degrees of Freedom Upper Bound}
\label{section:SDOFConverse}
In this section, we derive an SDOF upper bound for the $M \times K$ XNCM regardless of channel fading variations. We first present the following lemma, which will be used as an important intermediate step for the proof of the SDOF upper bound.

\begin{lemma}[Role of a Helper in X Networks]
\label{lemma:MacBound}
For any $\hat{k} \in \mathcal{K}$ and $p \in \mathcal{M}$, we have the following bound for the secrecy rate:
\begin{align} \label{Eqn:MacBound}
    n \sum_{j \in \mathcal{M}-p} R_{\hat{k} j} + h(X_{p}^{n} + \tilde{N}^{n}) \leq h(Y_{\hat{k}}^{n}) + nO(1),
\end{align}	
where the lowercase letter $h$ represents the differential entropy, and $\tilde{N}^{n}$ is an $n \times 1$ vector with independent Gaussian entries $\tilde{n}(t)$ ($t=[1:n]$) with respective variances less than $\frac{1}{|h_{\hat{k}p}(t)|^{2}}$.
\end{lemma}
\begin{proof}
	The proof follows the same line as that of the \emph{Role of a Helper} Lemma in \cite{Jianwei2012}, with the details in Appendix \ref{Appendix:Proof:LemmaMac}.
\end{proof}
\emph{Lemma} \ref{lemma:MacBound} states that in order to decode the messages $\mathbf{W}_{\hat{k},\mathcal{M}-p}$ at receiver $\hat{k}$, the differential entropy of transmitted signal of the transmitter $p$ should be upper bounded by the difference of the differential entropy at the receiver $\hat{k}$ and the sum rate of $\mathbf{W}_{\hat{k},\mathcal{M}-p}$. Now we are ready to present the SDOF upper bound for the $M \times K$ XNCM. 

\begin{theorem} \label{theorem:SDOFUB}
The optimal sum SDOF of the $M \times K$ XNCM is upper bounded as $d \leq \frac{K(M-1)}{K+M-2}$.
\end{theorem}
\begin{proof}
The detailed proof follows similarly to the SDOF upper bound proof for the interference channel in \cite{Jianwei2013}. For brevity, we define the following notation. Let $\mathbf{W}_{\mathcal{I},\mathcal{J}} = \{W_{i,j} | i \in \mathcal{I}, j \in \mathcal{J}\}$, with two finite sets $\mathcal{I} = \{1,2,\dots, I\}$ and $\mathcal{J} = \{1,2,\dots,J\}$. Let $\mathbf{X}_{\mathcal{J}} = \{X_{j}^{n} | j \in \mathcal{J}\}$, $\mathbf{Y}_{\mathcal{I}} = \{Y_{i}^{n} | i \in \mathcal{I}\}$, and $\mathbf{N}_{\mathcal{I}} = \{N^{n}_{i} | i \in \mathcal{I}\}$ denote the transmitted, received signal and the noise sequences in the set $\mathcal{J}$ and $\mathcal{I}$, respectively. Consider all messages in the network that are confidential for receiver $\hat{k}$. Starting from Fano's inequality, we have
\begin{align} \nonumber
    & ~~~n \sum_{i \in \mathcal{K}-\hat{k}, j \in \mathcal{M}} R_{ij} = H(\mathbf{W}_{\mathcal{K}-\hat{k}, \mathcal{M}}) \\ \nonumber 
    		& = I(\mathbf{W}_{\mathcal{K}-\hat{k}, \mathcal{M}}; \mathbf{Y}_{\mathcal{K}-\hat{k}}) + H(\mathbf{W}_{\mathcal{K}-\hat{k}, \mathcal{M}} | \mathbf{Y}_{\mathcal{K}-\hat{k}}) \\ \nonumber
        & \leq I(\mathbf{W}_{\mathcal{K}-\hat{k}, \mathcal{M}}; \mathbf{Y}_{\mathcal{K}-\hat{k}}) + n\epsilon_{1} \\ \nonumber
        & \leq I(\mathbf{W}_{\mathcal{K}-\hat{k}, \mathcal{M}}; \mathbf{Y}_{\mathcal{K}-\hat{k}}) - I(\mathbf{W}_{\mathcal{K}-\hat{k}, \mathcal{M}}; \mathbf{Y}_{\hat{k}}) \\ \label{Eqn:SecConstr}
        & ~~~~ + n(\epsilon_{1} + \epsilon_{2}) \\ \nonumber
        & \leq I(\mathbf{W}_{\mathcal{K}-\hat{k}, \mathcal{M}}; \mathbf{Y}_{\mathcal{K}}) - I(\mathbf{W}_{\mathcal{K}-\hat{k}, \mathcal{M}}; \mathbf{Y}_{\hat{k}}) + n\epsilon_{3} \\ \nonumber
        & = I(\mathbf{W}_{\mathcal{K}-\hat{k}, \mathcal{M}}; \mathbf{Y}_{\mathcal{K}-\hat{k}} | \mathbf{Y}_{\hat{k}}) + n\epsilon_{3} \\ \nonumber
        & = h(\mathbf{Y}_{\mathcal{K}-\hat{k}} | \mathbf{Y}_{\hat{k}}) - h(\mathbf{Y}_{\mathcal{K}-\hat{k}} | \mathbf{Y}_{\hat{k}}, \mathbf{W}_{\mathcal{K}-\hat{k},\mathcal{M}}) + n\epsilon_{3} \\ \label{Eqn:Ub1}
        & \leq h(\mathbf{Y}_{\mathcal{K}-\hat{k}} | \mathbf{Y}_{\hat{k}}) + n\epsilon_{3} \\ \nonumber 
        & = h(\mathbf{Y}_{\mathcal{K}}) - h(\mathbf{Y}_{\hat{k}}) + n\epsilon_{3} \\ 
        \label{Eqn:IntroX}
        & = h(\tilde{\mathbf{X}}_{\mathcal{M}},\mathbf{Y}_{\mathcal{K}}) - h(\tilde{\mathbf{X}}_{\mathcal{M}}|\mathbf{Y}_{\mathcal{K}}) - h(\mathbf{Y}_{\hat{k}}) + n\epsilon_{3} \\ \nonumber
        & \leq h(\tilde{\mathbf{X}}_{\mathcal{M}}) + h(\mathbf{Y}_{\mathcal{K}}|\tilde{\mathbf{X}}_{\mathcal{M}}) - h(\tilde{\mathbf{X}}_{\mathcal{M}}|\mathbf{Y}_{\mathcal{K}}, \mathbf{X}_{\mathcal{M}}) - h(\mathbf{Y}_{\hat{k}}) + n\epsilon_{3} \\ 
        \label{Eqn:Ub3}
        & \leq h(\tilde{\mathbf{X}}_{\mathcal{M}}) - h(\mathbf{Y}_{\hat{k}}) + no(\textrm{log}(P))
\end{align}        
%        \nonumber
%		& = I(\mathbf{Y}_{\mathcal{K}};\mathbf{W}_{\mathcal{K},\mathcal{M}}) + h(\mathbf{Y}_{\mathcal{K}}|\mathbf{W}_{\mathcal{K},\mathcal{M}}) - h(\mathbf{Y}_{\hat{k}}) + n\epsilon_{3} \\ \nonumber
%		& = H(\mathbf{W}_{\mathcal{K},\mathcal{M}}) - H(\mathbf{W}_{\mathcal{K},\mathcal{M}}|\mathbf{Y}_{\mathcal{K}}) + h(\mathbf{Y}_{\mathcal{K}}|\mathbf{W}_{\mathcal{K},\mathcal{M}}) \\ \nonumber 
%		& ~~~~ - h(\mathbf{Y}_{\hat{k}}) + n\epsilon_{3}  \\ \nonumber
%        & \leq H(\mathbf{W}_{\mathcal{K},\mathcal{M}}) + h(\mathbf{Y}_{\mathcal{K}}|\mathbf{W}_{\mathcal{K},\mathcal{M}}) - h(\mathbf{Y}_{\hat{k}}) + n\epsilon_{3} \\ \label{Eqn:Ub2}
%	    & = H(\mathbf{W}_{\mathcal{K},\mathcal{M}}) + h(\mathbf{Y}_{\mathcal{K}}|\mathbf{W}_{\mathcal{K},\mathcal{M}}, \mathbf{X}_{\mathcal{M}}) - h(\mathbf{Y}_{\hat{k}}) + n\epsilon_{3} \\ \nonumber
%	    & = H(\mathbf{W}_{\mathcal{K},\mathcal{M}}) + h(\mathbf{N}_{\mathcal{K}}|\mathbf{W}_{\mathcal{K},\mathcal{M}}, \mathbf{X}_{\mathcal{M}}) - h(\mathbf{Y}_{\hat{k}}) + n\epsilon_{3} \\ \label{Eqn:Ub3}
%        & = H(\mathbf{W}_{\mathcal{K},\mathcal{M}}) - h(\mathbf{Y}_{\hat{k}}) + n\epsilon_{4}
for some $\epsilon_{l} > 0$ and $\epsilon_{l} = o(\textrm{log}(P))$ ($l = 1,2,3$), where $\epsilon_{3} = \epsilon_{1} + \epsilon_{2}$.
\begin{itemize}
	\item (\ref{Eqn:SecConstr}) follows from the secrecy constraint
		\begin{align}\nonumber
		I(\mathbf{W}_{\mathcal{K}-\hat{k}, \mathcal{M}}; \mathbf{Y}_{\hat{k}}) \leq n\epsilon_{2},~		\textrm{for some}~ \epsilon_{2} = o(\textrm{log}(P)).
		\end{align}
	\item (\ref{Eqn:Ub1}) follows from
		\begin{align} \nonumber
    & ~~h(\mathbf{Y}_{\mathcal{K}-\hat{k}} | \mathbf{Y}_{\hat{k}}, \mathbf{W}_{\mathcal{K}-\hat{k},\mathcal{M}}) \\ \nonumber
    & \geq h(\mathbf{Y}_{\mathcal{K}-\hat{k}} | \mathbf{Y}_{\hat{k}}, \mathbf{W}_{\mathcal{K}-\hat{k},\mathcal{M}}, \mathbf{W}_{\hat{k},\mathcal{M}}, \mathbf{X}_{\mathcal{M}}) \\ \nonumber 
    & = h(\mathbf{N}_{\mathcal{K}-\hat{k}}) \geq 0,
		\end{align}
where the last inequality follows from the fact that the Gaussian distribution with unit variance has non-negative differential entropy.
	\item (\ref{Eqn:IntroX}) follows by defining $\tilde{X}^{n}_{i} = X^{n}_{i} + \tilde{N}_{i}^{n}$ for $i \in \mathcal{M}$, where $\tilde{N}_{i}^{n}$ is an $n \times 1$ vector with independent Gaussian entries $\tilde{n}_{i}(t)$ ($t=[1:n]$) with respective variances smaller than $\frac{1}{|h_{\hat{k}i}(t)|^{2}}$. 
	\item (\ref{Eqn:Ub3}) follows from the fact that 
	\begin{align}
	\nonumber 
	& ~~~|h(\mathbf{Y}_{\mathcal{K}}|\tilde{\mathbf{X}}_{\mathcal{M}}) - h(\tilde{\mathbf{X}}_{\mathcal{M}}|\mathbf{Y}_{\hat{k}}, \mathbf{X}_{\mathcal{M}})| \\ \nonumber
	& \leq |h(\mathbf{Y}_{\mathcal{K}}|\tilde{\mathbf{X}}_{\mathcal{M}})| \!+\! |h(\tilde{\mathbf{X}}_{\mathcal{M}}|\mathbf{Y}_{\hat{k}}, \mathbf{X}_{\mathcal{M}})| \leq n\delta,
	\end{align}
	for some $\delta = o(\textrm{log}(P))$, because only effective noise contributes to the differential entropy terms by conditioning on the signals $\mathbf{X}_{\mathcal{M}}$ and $\tilde{\mathbf{X}}_{\mathcal{M}}$, and the channel coefficients have bounded absolute values. 
\end{itemize}

So far, we have bounded the sum rate of all confidential messages that are not intended for receiver $\hat{k}$ as (\ref{Eqn:Ub3}). In order to bound the sum rate for the whole message set $\mathbf{W}_{\mathcal{K},\mathcal{M}}$, we would like to apply \emph{Lemma} \ref{lemma:MacBound} in the proof. We begin with (\ref{Eqn:Ub3}) as follows:
\begin{align} \nonumber
    & ~~~ n \sum_{i \in \mathcal{K}-\hat{k},j \in \mathcal{M}} R_{ij} \\ \nonumber
    & \leq h(\tilde{\mathbf{X}}_{\mathcal{M}}) - h(\mathbf{Y}_{\hat{k}}) + no(\textrm{log}(P)) \\ \nonumber
	& \leq \sum_{p\in \mathcal{M}} h(\tilde{\mathbf{X}}_{p}) - h(\mathbf{Y}_{\hat{k}}) + no(\textrm{log}(P)) \\
	\label{Eqn:SubMacArg}
    & \leq \sum_{p\in \mathcal{M}} \left(h(\mathbf{Y}_{\hat{k}}) - n\sum_{j \in \mathcal{M}-p} R_{\hat{k}j} \right) - h(\mathbf{Y}_{\hat{k}}) + no(\textrm{log}(P)),
\end{align}
where (\ref{Eqn:SubMacArg}) follows by substituting (\ref{Eqn:MacBound}) of \emph{Lemma} \ref{lemma:MacBound}. Equivalently, we have
\begin{align} \nonumber
	& n\sum_{i \in \mathcal{K}-\hat{k},j \in \mathcal{M}} R_{ij} + n\sum_{p \in \mathcal{M}} \sum_{j \in \mathcal{M}-p} R_{\hat{k}j} \\ \nonumber
	& ~~~\leq (M-1) h(\mathbf{Y}_{\hat{k}}) + no(\textrm{log}(P)).
\end{align}
Manipulating the indices on the left hand side of the above inequality, we have
\begin{align} \nonumber
	& n \sum_{i \in \mathcal{K}, j \in \mathcal{M}} R_{ij} + n(M-2) \sum_{j \in \mathcal{M}}R_{\hat{k}j} \\ \nonumber
	& ~~~ \leq (M-1)h(\mathbf{Y}_{\hat{k}}) + no(\textrm{log}(P)).
\end{align}
Considering that $h(\mathbf{Y}_{\hat{k}}) \leq n \textrm{log}(P) + nO(1)$, we sum up the above inequality for all $\hat{k} \in \mathcal{K}$ to obtain:
\begin{align} \nonumber
	(K+M-2) \sum_{i\in \mathcal{K}, j\in \mathcal{M}}R_{ij} \leq K(M-1)\textrm{log}(P) + o(\textrm{log}(P)).
\end{align}
Therefore, we have the SDOF upper bound as follows to conclude the proof
\begin{equation} \nonumber
    d = \sum_{i \in \mathcal{K}, j \in \mathcal{M}} d_{ij} \leq \frac{K(M-1)}{M+K-2}.
\end{equation} 	
\end{proof}

\begin{remark} \rm
We note that the derived SDOF upper bound does not make any assumptions regarding channel fading variations. Therefore, \emph{Theorem} \ref{theorem:SDOFUB} provides a general SDOF upper bound for the $M \times K$ XNCM with constant channels, or with time/frequency varying channels. We also note that the derived bound naturally serves as an upper bound for the SDOF of the considered network with no CSIT. 
\end{remark}

\begin{remark}[At most half of the cake] \rm
We observe that the derived SDOF upper bound for the $M \times K$ XNCM equals to the sum degrees of freedom (DOF) of the $(M-1) \times K$ X network without secrecy constraints. Therefore, the impact of confidentiality is equivalent to removing at least one sender from the wireless X network, in terms of DOF. It is also interesting to note that when $K=M$, the above upper bound yields $\sum_{i,j \in \mathcal{K}} d_{ij} \leq \frac{K(K-1)}{2(K-1)} = \frac{K}{2}$, which coincides with the sum DOF for the $K$-user interference channel without secrecy constraints. Therefore, for the $K \times K$ fully connected wireless $X$ network, if all the messages existing in the network are confidential, then every sender can at most obtain half of the resources.
\end{remark}

\section{Secure Degrees of Freedom of the $M \times K$ XNCM with Time/Frequency-Varying Channels}
\label{section:SDOFXNCM}
In this section, we study the SDOF of wireless X networks with time or frequency varying channels. We show that the proposed SDOF upper bound is tight when $K = 2$ by presenting an artificial noise alignment scheme. The achieved SDOF of XNCM with an external eavesdropper 
will also be studied in this section. We start with the following theorem.

\begin{figure*}[t!]
\centerline{\epsfig{file=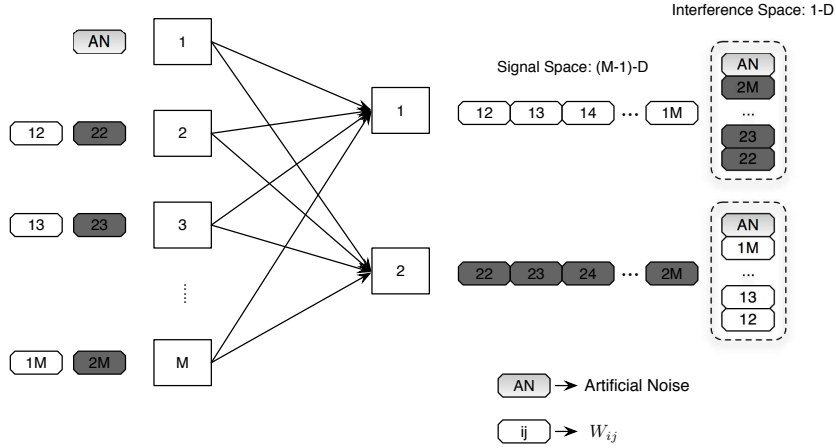,width=110mm}}
\caption{Artificial noise alignment (ANA) for $M \times 2$ XNCM.}
\label{Fig:XNC2ANA}
\end{figure*}

\begin{theorem}
\label{theorem:SDOFXNCM}
The optimal sum SDOF of the $M \times K$ XNCM with time/frequency-varying channels is 
\begin{align} \nonumber
	d = \frac{K(M-1)}{K+M-2}, ~\textrm{if}~K=2,
\end{align}
and
\begin{align} \nonumber
	\frac{K(M-1)}{K+M-1} \leq d \leq \frac{K(M-1)}{K+M-2}, ~\textrm{if}~K \geq 3.
\end{align}
\end{theorem} 

\begin{proof}
The converse follows directly from \emph{Theorem} \ref{theorem:SDOFUB}. For achievability, we start from the case $K = 2$. We shall show that SDOF of $\frac{2(M-1)}{M}$ can be achieved. In the following, we propose an ANA approach, which essentially aligns the artificial noise to the interference space at the receivers. The details are presented as follows.

Considering an $M$ symbol extension over the original channel, we will show that a total of $2(M-1)$ SDOF can be achieved. The main idea is to treat one transmitter, say transmitter $1$, as a special sender, which transmits only artificial noise. For the rest of the $(M-1)$ transmitters, each sends two confidential messages intended for two receivers. By aligning the artificial noise and interference in the same subspace, the information leakage can be bounded. We illustrate the alignment design in Fig. \ref{Fig:XNC2ANA}. Therefore, for message $W_{k,1}$, $k \in {1,2}$, the rate $R_{k1} = 0$. For transmitters $2$ to $M$, we shall show that each confidential message $W_{k,j}$ ($j \in \mathcal{M}-1$), has one SDOF over the channel extensions.
Before presenting the detailed transmission scheme, let us first provide the channel input-output relation for the $M$-symbol extension as follows:
\begin{align}
	\nonumber
	Y_{k}(t) = \sum_{m\in \mathcal{M}} \mathbf{H}_{km}(t)X_{m}(t) + N_{k}(t), ~ k\in \mathcal{K},
\end{align}
where $\mathbf{H}_{km}(t)$ is an $M \times M$ diagonal matrix with diagonal elements $h_{km}(t(M-1)+i)$ for $i=[1:M]$. The terms $X_{m}(t), N_{k}(t)$ and $Y_{k}(t)$ are the $M \times 1$ input, noise and output vectors, respectively. In the following, the time slot $t$ is omitted for simplicity. 

The transmitted signal of transmitter $1$ is
\begin{align} \nonumber
	X_{1} & = \boldsymbol{\Phi}^{[1]} \nu,
\end{align}
where $\nu$ is an artificial noise symbol chosen from $\mathcal{CN}(0,P)$, and $\boldsymbol{\Phi}^{[1]}$ is an $M \times 1$ beamforming matrix. Note that we enforce the artificial noise symbol to be Gaussian, as it can provide the maximum differential entropy to confuse the eavesdropper. At the other transmitters, the transmitted signals can be written as
\begin{align} \nonumber
	X_{j} & = \sum_{i=1,2} \boldsymbol{\Phi}^{[ij]} \mu_{ij}, ~j \in \mathcal{M}, j \neq 1,
\end{align}
where $\mu_{ij}$ is the confidential message symbol originating from transmitter $j$ to receiver $i$, and the beamforming matrix $\boldsymbol{\Phi}^{[ij]}$ has dimension $M \times 1$. Then, the received signal at the receiver $k$ is ($l \neq k$)
\begin{align} \nonumber
	Y_{k} & = \mathbf{H}_{k1}\boldsymbol{\Phi}^{[1]} \nu + \sum_{j=2}^{M} \mathbf{H}_{kj} \left(\sum_{i=1,2}\boldsymbol{\Phi}^{[ij]} \mu_{ij} \right) + N_{k} \\ \nonumber
	& = \sum_{j=2}^{M} \mathbf{H}_{kj} \boldsymbol{\Phi}^{[kj]} \mu_{kj} + \sum_{j=2}^{M} \mathbf{H}_{kj} \boldsymbol{\Phi}^{[lj]} \mu_{lj} + \mathbf{H}_{k1} \boldsymbol{\Phi}^{[1]} \nu + N_{k},
\end{align}
where $\mathbf{H}_{kj}$ is an $M \times M$ diagonal matrix with each diagonal element chosen independently from a continuous distribution. In the achievable scheme, we aim to perfectly align the interference with artificial noise. Therefore, at receiver $1$, we have
\begin{align} \label{Eqn:BFRelation1}
	\mathbf{H}_{1j}\boldsymbol{\Phi}^{[2j]} = \mathbf{H}_{11}\boldsymbol{\Phi}^{[1]}, ~ j \in \mathcal{M}-1,
\end{align}
and similarly, at receiver $2$, we have
\begin{align} \label{Eqn:BFRelation2}
	\mathbf{H}_{2j}\boldsymbol{\Phi}^{[1j]} = \mathbf{H}_{21}\boldsymbol{\Phi}^{[1]}, ~ j \in \mathcal{M}-1.
\end{align}
Thus, we have
\begin{align} \nonumber
	\boldsymbol{\Phi}^{[kj]} = \mathbf{H}_{ij}^{-1}\mathbf{H}_{i1}\boldsymbol{\Phi}^{[1]}, j \in \mathcal{M}-1, k,i \in \{1,2\}, k \neq i.
\end{align}
As the next step, we have to show the effective channel matrices at receivers are full rank. Let $\boldsymbol{\Lambda}^{[k]}$ represent the effective channel matrix at receiver $k$, where
\begin{align} \nonumber
	\boldsymbol{\Lambda}^{[k]} & = \left[\mathbf{H}_{k2}\boldsymbol{\Phi}^{[k2]}~\mathbf{H}_{k3}\boldsymbol{\Phi}^{[k3]}~ \cdots~ \mathbf{H}_{kM}\boldsymbol{\Phi}^{[kM]}~ \mathbf{H}_{k1}\boldsymbol{\Phi}^{[1]} \right].
\end{align}
Let $\boldsymbol{\Phi}^{[1]} = [\phi_{1}~\phi_{2}~ \cdots~ \phi_{M}]^{T}$,
%\begin{align} \nonumber
%%	\boldsymbol{\Phi}^{[1]} = \left[\begin{array}{c} \phi_{1} \\ \phi_{2} \\ \vdots \\ \phi_{M}\end{array}\right],
%	\boldsymbol{\Phi}^{[1]} = [\phi_{1}~\phi_{2}~ \cdots~ \phi_{M}]^{T},
%\end{align}
where each element $\phi_{m}$ ($m \in \mathcal{M}$) is chosen independently from a continuous distribution. With the relations (\ref{Eqn:BFRelation1}) and (\ref{Eqn:BFRelation2}), the effective channel matrices at receivers are presented in the bottom of the next page.
\begin{figure*}[bp]
\hrulefill
\begin{align} \nonumber
	\boldsymbol{\Lambda}^{[1]} & = \left[\mathbf{H}_{12}\mathbf{H}_{22}^{-1}\mathbf{H}_{21}\boldsymbol{\Phi}^{[1]}~~ \mathbf{H}_{13}\mathbf{H}_{23}^{-1}\mathbf{H}_{21}\boldsymbol{\Phi}^{[1]}~~ \cdots~~ \mathbf{H}_{1M}\mathbf{H}_{2M}^{-1}\mathbf{H}_{21}\boldsymbol{\Phi}^{[1]}~~ \mathbf{H}_{11}\boldsymbol{\Phi}^{[1]}\right]	\\ \nonumber
	\boldsymbol{\Lambda}^{[2]} & = \left[\mathbf{H}_{22}\mathbf{H}_{12}^{-1}\mathbf{H}_{11}\boldsymbol{\Phi}^{[1]}~~ \mathbf{H}_{23}\mathbf{H}_{13}^{-1}\mathbf{H}_{11}\boldsymbol{\Phi}^{[1]}~~
\cdots~~ \mathbf{H}_{2M}\mathbf{H}_{1M}^{-1}\mathbf{H}_{11}\boldsymbol{\Phi}^{[1]}~~ \mathbf{H}_{21}\boldsymbol{\Phi}^{[1]} \right]
\end{align}
\end{figure*}
In order to show that $\boldsymbol{\Lambda}^{[k]}$ has full rank almost surely, we will equivalently show that $|\boldsymbol{\Lambda}^{[k]}| \neq 0$ with probability one. We adopt the method given in the proof of Lemma $1$ in \cite{Cadambe2009}. Let $\lambda^{[k]}_{ij}$ represent the element in the $i$-th row and $j$-th column of $\boldsymbol{\Lambda}^{[k]}$. We observe that every $\lambda^{[k]}_{ij}$ can be written in the following form: $\lambda_{ij}^{[k]} = \prod_{q=1}^{Q} \left(\beta_{i}^{[q]}\right)^{\alpha_{ij}^{[q]}}$, where $\beta_{i}^{[q]}$ is a random variable and all exponents are integers, $\alpha_{ij}^{[q]} \in \mathbb{Z}$. Furthermore,
\begin{itemize}
	\item $\beta_{i}^{[q]} | \{ \beta_{i^{'}}^{[q^{'}]}, \forall (i,q) \neq (i^{'}, q^{'}) \}$ has a continuous cumulative probability distribution, and
	\item $\forall i,j,j^{'} \in \mathcal{M}$ and $j \neq j^{'}$
	\begin{equation}
		\nonumber
		\left(\alpha_{ij}^{[1]}, \alpha_{ij}^{[2]}, \dots, \alpha_{ij}^{[Q]} \right) \neq \left( \alpha_{ij^{'}}^{[1]}, \alpha_{ij^{'}}^{[2]}, \dots, \alpha_{ij^{'}}^{[Q]}\right).
	\end{equation}
\end{itemize}
Let $C_{ij}^{[k]}$ represent the cofactor corresponding to $\lambda_{ij}^{[k]}$. Then
\begin{equation} \nonumber
	|\boldsymbol{\Lambda}^{[k]}| = \lambda_{11}^{[k]}C_{11}^{[k]} + \lambda_{12}^{[k]}C_{12}^{[k]} + \dots + \lambda_{1M}^{[k]}C_{1M}^{[k]}.
\end{equation}
$|\boldsymbol{\Lambda}^{[k]}| = 0$ with nonzero probability only if at least one of the following two conditions is satisfied:
\begin{itemize}
	\item $\beta_{i}^{[q]}$, $q = 1, 2, \dots, Q$ are roots of the polynomial formed by setting $|\boldsymbol{\Lambda}^{[k]}| = 0$.
	\item The polynomial is the zero polynomial.
\end{itemize}
Note that $\beta_{i}^{[q]}$ has a continuous cumulative joint distribution conditioned on $C_{1l}$, $l \in \mathcal{M}$. Therefore, the probability of $\beta_{i}^{[q]}$ taking values from a finite root set goes to zero. Thus, the first condition collapses almost surely. For the second condition, since each $\beta_{i}^{[q]}$ has a unique set of exponents, the argument of the zero polynomial holds with positive probability only if $\textrm{Pr}(C_{1l}^{[k]} = 0) > 0$, for $l \in \mathcal{M}$. Because $C_{1l}^{[k]}$ is also the determinant of a submatrix of $\boldsymbol{\Lambda}^{[k]}$, the same argument can be iteratively performed, until reaching to a single element matrix containing $\lambda_{Ml}$. Therefore,
\begin{equation} \nonumber
	\textrm{Pr}(|\boldsymbol{\Lambda}^{[k]}| = 0) > 0 \Rightarrow \textrm{Pr}(|\lambda_{Ml}^{[k]}| = 0) > 0.
\end{equation}
Since $\lambda_{Ml}^{[k]}$ has the form of products of continuous random variables, we can conclude that $\textrm{Pr}(|\lambda_{Ml}^{[k]}| = 0) = 0$. Therefore $|\boldsymbol{\Lambda}^{[k]}| \neq 0$ almost surely, and $\boldsymbol{\Lambda}^{[k]}$ is of full rank almost surely.

Because the desired signal occupies $M-1$ dimensions of the $M$-symbol extension in $\boldsymbol{\Lambda}^{[k]}$, we can write the sum rate of confidential messages $W_{kj}$ as follows:
\begin{align}
	\nonumber
	\sum_{j=2}^{M} R_{kj} = \frac{M-1}{M} \textrm{log}(P) + o(\textrm{log}(P)), k = 1,2.
\end{align}
In the following, we compute the rate equivocation after collecting the whole codeword. By choosing $\mathcal{S}_{\mathcal{I}} = {2}$, $\mathcal{S}_{\mathcal{J}} = \mathcal{M}$, the equivocation at receiver $1$ is
\begin{align}
	\nonumber
	& ~ \Delta_{\mathcal{S}_{\mathcal{I}}, \mathcal{S}_{\mathcal{J}}}^{[1]} = \Delta_{2,\mathcal{M}} = 1 - \frac{I(\mathbf{W}_{2,\mathcal{M}}; Y_{1}^{n})}{n\sum_{j \in \mathcal{M}}R_{2j}} \\ \nonumber
	& \geq 1 - \frac{ \frac{n}{M} I(\boldsymbol{\mu}_{2,\mathcal{M}}; Y_{1})}{n\sum_{j \in \mathcal{M}}R_{2j}} \\ \nonumber
	& = 1 - \frac{ \frac{n}{M} \left(I(\boldsymbol{\mu}_{2, \mathcal{M}}, \boldsymbol{\nu}; Y_{1}) - I(\boldsymbol{\nu}; Y_{1} | \boldsymbol{\mu}_{2, \mathcal{M}})\right)}{n\sum_{j \in \mathcal{M}}R_{2j}},
\end{align}
by the data processing inequality and the fact that channels are memoryless, where $\boldsymbol{\mu}_{k,\mathcal{M}} = \{\mu_{k,j}, j \in \mathcal{M}\}$, with $\mu_{k,1} \in \emptyset$. We next bound the mutual information terms
\begin{align}
	\nonumber
	& ~ I(\boldsymbol{\mu}_{2, \mathcal{M}}, \boldsymbol{\nu}; Y_{1}) \leq I(\boldsymbol{\mu}_{2, \mathcal{M}}, \boldsymbol{\nu}; Y_{1}, \boldsymbol{\mu}_{1,\mathcal{M}}) \\ \nonumber
	& = I(\boldsymbol{\mu}_{2, \mathcal{M}}, \boldsymbol{\nu}; Y_{1} | \boldsymbol{\mu}_{1,\mathcal{M}}) \\ \nonumber
	& = h(Y_{1} | \boldsymbol{\mu}_{1,\mathcal{M}}) - h(Y_{1} | \boldsymbol{\mu}_{1,\mathcal{M}}, \boldsymbol{\mu}_{2, \mathcal{M}}, \boldsymbol{\nu}) = \textrm{log}(P) - O(1).
\end{align}
Similarly, $I(\boldsymbol{\nu}; Y_{1} | \boldsymbol{\mu}_{2, \mathcal{M}})  = \textrm{log}(P) + o(\textrm{log}(P)).$
Therefore, we can show the equivocation
\begin{equation}
	\nonumber
	\lim_{n,P \rightarrow \infty} \Delta_{2,\mathcal{M}}^{[1]} = 1.
\end{equation}
A similar argument can also be applied to receiver $2$. Overall, the sum SDOF $\frac{2(M-1)}{M}$ can be achieved for the $M \times 2$ XNCM.

For $K \geq 3$, the SDOF lower bound overlaps with the results in \cite{Tiangao2008}, where random binning is used to provide secrecy. As we will show in the proof of \emph{Lemma} \ref{lemma:SDOFXNCMEE}, this lower bound can also be achieved by an ANA scheme. It is worth noting that the proposed scheme collapses if we try to apply it directly for $K \geq 3$, in particular to use $K-1$ dimensions of AN from the helper. The problem is that the independence of the signal and interference space will be violated. A further study needs to be carried out for closing the SDOF gap between the upper and lower bound. In \cite{ZhaoThesis2015}, a toy example is presented to discuss the limitation of the current scheme for achieving the derived upper bound.
\end{proof}

\begin{remark}[Time sharing of transmitters] \rm
We note that in the above proof, the optimal sum SDOF is essentially achieved by treating one transmitter as a helper. In order to achieve a symmetric secrecy rate for each transmitter, a time sharing protocol can be applied, in which each transmitter take turns to be the helper in a block-fashion transmission. By this means, e.g., when $K=2$, each transmitter can obtain SDOF of $\frac{2(M-1)}{M^{2}}$.
\end{remark}

\begin{remark}[Relation to broadcast channels] \rm
As shown in the proof, for $M = 2$, the optimal sum SDOF of the $2 \times 2$ XNCM is shown to be $1$, which coincides with the optimal SDOF of the single-input single-output (SISO) broadcast channel with confidential messages if there exists an additional helper in the network, as shown in \cite{Jianwei2012}.
\end{remark}

\begin{remark}[Relation to the X channel with feedback] \rm
It has been shown in \cite{Zaidi2013} that when $M=K=2$ if the considered network has output feedback and delayed CSIT, the optimal sum SDOF is also $1$, which ties with our result here. Therefore, for a $2 \times 2$ X network, the output feedback and delayed CSIT can be as good as providing perfect CSIT in terms of degrees of freedom. 
\end{remark}

In the following, we investigate the SDOF of the $M \times K$ XNCM with an external eavesdropper (XNCM-EE). The achieved SDOF of the considered network also implies the lower bound in \emph{Theorem} \ref{theorem:SDOFXNCM} when $K \geq 3$, because removing the eavesdropper will not decrease the secrecy rate. We present the results in the following lemma. 

\begin{lemma}
\label{lemma:SDOFXNCMEE}
For the $M \times K$ XNCM-EE with time/frequency varying channels, the optimal sum SDOF can be bounded as $\frac{K(M-1)}{K+M-1} \leq d \leq \frac{K(M-1)}{K+M-2+\frac{1}{M}}$. 	
\end{lemma}

\begin{proof}
For the converse, let us first remove the secrecy constraints at all $K$ receivers, which certainly will not harm the secrecy capacity region. Then the considered network can be seen as an $M \times K$ X network with only an external eavesdropper, the sum SDOF of which can be bounded above by $\frac{K(M-1)}{K+M-2+\frac{1}{M}}$ based on Appendix \ref{Appendix:Proof:UB_XNCMEE}. Therefore it also serves as an SDOF upper bound for that of the considered $M \times K$ XNCM-EE.  

The detailed proof for the lower bound is presented as follows. Let $\Gamma = K(M-1)$. We will show that the SDOF $d = \frac{K(M-1)n^{\Gamma}}{K(n+1)^{\Gamma} + (M-1)n^{\Gamma}}$ can be achieved for any $n \in \mathbb{N}$, which yields $d = \frac{K(M-1)}{K+M-1}$ when taking the supremum for all $n$. Let $\mu_{n} = K(n+1)^{\Gamma} + (M-1)n^{\Gamma}$. We consider a $\mu_{n}$ symbol extension over the time-varying channel. Then, the channel input-output relationship is
\begin{align} \nonumber
    Y_{k} = \sum_{i \in \mathcal{M}}\mathbf{H}_{ki} X_{i} + N_{k}, ~ \forall k \in \{\mathcal{K},e\}, 
\end{align}
where $\mathbf{H}_{ki}$ is a $\mu_{n} \times \mu_{n}$ diagonal matrix, and $X_{i}$ is the $\mu_{n} \times 1$ signal vector from transmitter $i$.

Our essential idea is to let one specific transmitter, say $1$, send artificial noise only, while the other transmitter sends confidential messages to the intended receiver. Meanwhile we propose an interference alignment scheme to align the confidential messages with artificial noise at every unintended receiver and also the external eavesdropper, to avoid information leakage. For transmitters $2$ to $M$, we can design the transmitted signal as follows:
\begin{align} \nonumber
    X_{i} = \sum_{j=1}^{K} \boldsymbol{\Phi}^{[ji]} \boldsymbol{\mu}_{ji}, ~ \forall i \in \mathcal{M}-1,
\end{align}
where $\boldsymbol{\mu}_{ji}$ is the $n^{\Gamma} \times 1$ symbol vector coded from the confidential message $W_{ji}$, and $\boldsymbol{\Phi}^{[ji]}$ is the corresponding $\mu_{n} \times n^{\Gamma}$ beamforming matrix. At transmitter $M$, the signal can be specifically designed as
\begin{align} \nonumber
    X_{1} = \sum_{j=1}^{K} \boldsymbol{\Phi}^{[j1]} \boldsymbol{\nu}_{j1}, 
\end{align}
where $\boldsymbol{\nu}_{j1}$ is the $(n+1)^{\Gamma} \times 1$ artificial noise symbol vector chosen from Gaussian distribution $\mathcal{CN}(0, \frac{P}{(n+1)^{\Gamma}}\mathbf{I}_{(n+1)^{\Gamma}})$, and $\boldsymbol{\Phi}^{[ji]}$ is the corresponding $\mu_{n} \times (n+1)^{\Gamma}$ beamforming matrix. 
Thus, the received signal at the $k$-th receiver can be written as
\begin{align} \nonumber
    Y_{k} = \sum_{i=2}^{M}\mathbf{H}_{ki}\left(\sum_{j=1}^{K} \boldsymbol{\Phi}^{[ji]} \boldsymbol{\mu}_{ji} \right) + \mathbf{H}_{k1}\sum_{j=1}^{K} \boldsymbol{\Phi}^{[j1]} \boldsymbol{\nu}_{j1} + N_{k}.
\end{align}
In the following, we will introduce the details of the alignment. At receiver $k$, $k \in \mathcal{K}$, we would like to have the following alignment conditions, $\forall j \in \mathcal{K}-k$: 
\begin{align} \nonumber
\textrm{IA Block}~j&: \begin{cases} ~\mathbf{H}_{k2}\boldsymbol{\Phi}^{[j2]} \prec \mathbf{H}_{k1} \boldsymbol{\Phi}^{[j1]} \\ ~\mathbf{H}_{k3}\boldsymbol{\Phi}^{[j3]} \prec \mathbf{H}_{k1} \boldsymbol{\Phi}^{[j1]} \\ ~~~\vdots \\ ~\mathbf{H}_{kM}\boldsymbol{\Phi}^{[jM]} \prec \mathbf{H}_{k1} \boldsymbol{\Phi}^{[j1]}. \end{cases}
\end{align}
Generally speaking, there are $K-1$ alignment blocks at receiver $k$, and within each block we wish to align all the confidential messages intended for receiver $j$ to the subspace spanned by artificial noise $\boldsymbol{\nu}_{j1}$. Similarly, at the eavesdropper, we would like to have $K$ alignment blocks as follows, $\forall j \in \mathcal{K}$:
\begin{align} \nonumber
\textrm{IA Block}~j&: \begin{cases} ~\mathbf{H}_{e2}\boldsymbol{\Phi}^{[j2]} \prec \mathbf{H}_{e1} \boldsymbol{\Phi}^{[j1]} \\ ~\mathbf{H}_{e3}\boldsymbol{\Phi}^{[j3]} \prec \mathbf{H}_{e1} \boldsymbol{\Phi}^{[j1]} \\ ~~~\vdots \\ ~\mathbf{H}_{eM}\boldsymbol{\Phi}^{[jM]} \prec \mathbf{H}_{e1} \boldsymbol{\Phi}^{[j1]}. \end{cases}
\end{align}
Therefore, every confidential message intended for receiver $j \in \mathcal{K}$ is aimed to be aligned to the subspace of artificial noise $\boldsymbol{\Phi}^{[j1]}$ within the alignment block $j$. Let us collect the alignment block $j$ at every receiver, including the eavesdropper. All the relations can be written as
\begin{align} \nonumber
\begin{cases} ~\textrm{Receiver}~k \neq j: \begin{cases} ~\mathbf{H}_{k2}\boldsymbol{\Phi}^{[j2]} \prec \mathbf{H}_{k1} \boldsymbol{\Phi}^{[j1]} \\ ~\mathbf{H}_{k3}\boldsymbol{\Phi}^{[j3]} \prec \mathbf{H}_{k1} \boldsymbol{\Phi}^{[j1]} \\ ~~~\vdots \\ ~\mathbf{H}_{kM}\boldsymbol{\Phi}^{[jM]} \prec \mathbf{H}_{k1} \boldsymbol{\Phi}^{[j1]} \end{cases} \\
    ~\textrm{Eavesdroper}: \begin{cases} ~\mathbf{H}_{e2}\boldsymbol{\Phi}^{[j2]} \prec \mathbf{H}_{e1} \boldsymbol{\Phi}^{[j1]} \\ ~\mathbf{H}_{e3}\boldsymbol{\Phi}^{[j3]} \prec \mathbf{H}_{e1} \boldsymbol{\Phi}^{[j1]} \\ ~~~\vdots \\ ~\mathbf{H}_{eM}\boldsymbol{\Phi}^{[jM]} \prec \mathbf{H}_{e1} \boldsymbol{\Phi}^{[j1]}. \end{cases}
\end{cases}
\end{align}
Thus, there are $\Gamma = K(M-1)$ relations for alignment block $j$. To find the proper solution for all the beamforming matrices, we first let
\begin{align} \nonumber
    \boldsymbol{\Phi}^{[j2]} = \boldsymbol{\Phi}^{[j3]} = \cdots = \boldsymbol{\Phi}^{[jM]}.
\end{align}
Then all the $\Gamma$ relations can be written as
\begin{align} \nonumber
    \mathbf{T}^{[km]}\boldsymbol{\Phi}^{[j2]} \prec \boldsymbol{\Phi}^{[j1]}, ~ \forall k \in \{\mathcal{K}-j, e\}, ~m \in \mathcal{M}-1,
\end{align}
where
\begin{align} \nonumber
    \mathbf{T}^{[km]} = \left(\mathbf{H}_{k1}\right)^{-1} \mathbf{H}_{km}.
\end{align}
Reordering all the $\mathbf{T}^{[km]}$ by the index from $1$ to $\Gamma$, and following the method given in \cite{Cadambe2009}, $\boldsymbol{\Phi}^{[j1]}$ and $\boldsymbol{\Phi}^{[jM]}$ can be designed as
\begin{align}
    \nonumber
    \boldsymbol{\Phi}^{[j2]} & = \left\{ \left(\prod_{i=1,2,\dots, \Gamma} \left(\mathbf{T}^{[i]}\right)^{\alpha_{i}}\right) \mathbf{w}^{[j]} : \alpha_{i} \in \{1,2,\dots, n\} \right\} \\ \nonumber
    \boldsymbol{\Phi}^{[j1]} & = \left\{ \left(\prod_{i=1,2,\dots, \Gamma} \left(\mathbf{T}^{[i]}\right)^{\alpha_{i}}\right) \mathbf{w}^{[j]} : \alpha_{i} \in \{1,2,\dots, n+1\} \right\},
\end{align}
where $\mathbf{w}^{[j]}$ is the $\mu_{n} \times 1$ vector, with each element chosen independently from a continuous distribution with bounded absolute value. The same method can be applied for all $j \in \mathcal{K}$ alignment blocks. Therefore, the effective channel matrix at receiver $k$ can be written as
\begin{align} \nonumber
    \mathbf{C}_{k} & = \left[\mathbf{H}_{k1}\boldsymbol{\Phi}^{[k1]}~ \mathbf{H}_{k2}\boldsymbol{\Phi}^{[k2]}~ \cdots~ \mathbf{H}_{kM}\boldsymbol{\Phi}^{[kM]}~ \mathbf{I}_{k}\right] \\ \nonumber
        & = \left[\mathbf{H}_{k1}\boldsymbol{\Phi}^{[k1]}~ \mathbf{H}_{k2}\boldsymbol{\Phi}^{[k2]}~ \cdots~ \mathbf{H}_{kM}\boldsymbol{\Phi}^{[k2]}~ \mathbf{I}_{k}\right],    
\end{align}
with $\mathbf{I}_{k}$ defined on the top of next page.

% The definition of I_{k}
\begin{figure*}[!h]
\begin{align} \nonumber
    \mathbf{I}_{k} = \left[\mathbf{H}_{k1}\boldsymbol{\Phi}^{[11]}~ \mathbf{H}_{k1}\boldsymbol{\Phi}^{[21]}~ \cdots~ \mathbf{H}_{k1}\boldsymbol{\Phi}^{[(k-1),1]}~ \mathbf{H}_{k1}\boldsymbol{\Phi}^{[(k+1),1]}~ \cdots~ \mathbf{H}_{k1}\boldsymbol{\Phi}^{[K1]}~ \right].
\end{align}
\hrulefill
\end{figure*}
% End the definition

Following \emph{Lemma} $1$ and \emph{Lemma} $2$ in \cite{Cadambe2009}, we can prove that $\mathbf{C}_{k}$ is of full rank almost surely, where the signal occupies $(M-1)n^{\Gamma}$ independent dimensions. Therefore, we have
\begin{align} \nonumber
	\sum_{m \in \mathcal{M}} R_{km} = \frac{K(M-1)n^{\Gamma}}{\mu_{n}} \textrm{log}(P) + o(\textrm{log}(P)).
\end{align}

In the following, we would like to show that the information leakage can be bounded such that the secrecy constraints at the receivers and eavesdropper are satisfied. At the eavesdropper, we can write the received signal as follows:
\begin{align} \nonumber
    Y_{e} & = [\mathbf{H}_{e1}\boldsymbol{\Phi}^{[11]}~ \mathbf{H}_{e1}\boldsymbol{\Phi}^{[21]}~ \cdots~ \mathbf{H}_{e1}\boldsymbol{\Phi}^{[K1]}] \left[\begin{array}{c} \boldsymbol{\nu}_{11} \\ \boldsymbol{\nu}_{21} \\ \vdots \\ \boldsymbol{\nu}_{K1} \end{array}\right] \\ \nonumber 
    & ~~ + \sum_{i=2}^{M}[\mathbf{H}_{ei}\boldsymbol{\Phi}^{[1i]}\boldsymbol{\mu}_{1i}~ \mathbf{H}_{ei}\boldsymbol{\Phi}^{[2i]}\boldsymbol{\mu}_{2i}~ \cdots~ \mathbf{H}_{ei}\boldsymbol{\Phi}^{[Ki]}\boldsymbol{\mu}_{Ki}] + N_{e}.
\end{align}
With $\mathcal{S}_{I} = \mathcal{K}$ and $\mathcal{S}_{J}=\mathcal{M}-1$, the information leakage at the eavesdropper is
\begin{align}
\nonumber
& I(\mathbf{W}_{\mathcal{S}_{I},\mathcal{S}_{J}}; Y_{e}^{n}) \leq \frac{n}{\mu_{n}}I(\boldsymbol{\mu}_{\mathcal{S}_{I},\mathcal{S}_{J}};Y_{e}) \\ \nonumber
& ~~~ = \frac{n}{\mu_{n}} \left( I(\boldsymbol{\mu}_{\mathcal{S}_{I},\mathcal{S}_{J}}, \boldsymbol{\nu}_{\mathcal{S}_{I},1};Y_{e}) - I(\boldsymbol{\nu}_{\mathcal{S}_{I},1};Y_{e} | \boldsymbol{\mu}_{\mathcal{S}_{I},\mathcal{S}_{J}})\right).
\end{align}
Let $\mathbf{A} = [\mathbf{H}_{e1}\boldsymbol{\Phi}^{[11]}~ \mathbf{H}_{e1}\boldsymbol{\Phi}^{[21]}~ \cdots~ \mathbf{H}_{e1}\boldsymbol{\Phi}^{[K1]}]$ and $\mathbf{B}_{i} = [\mathbf{H}_{ei}\boldsymbol{\Phi}^{[1i]}\boldsymbol{\mu}_{1i}~ \mathbf{H}_{ei}\boldsymbol{\Phi}^{[2i]}\boldsymbol{\mu}_{2i}~ \cdots~ \mathbf{H}_{ei}\boldsymbol{\Phi}^{[Ki]}\boldsymbol{\mu}_{Ki}]$. Because of the alignment blocks, it is readily shown that $\textrm{span}(\mathbf{B}_{i}) \subseteq \textrm{span}(\mathbf{A})$. By \emph{Lemma} \ref{Theorem:lemma:RankInc} in Appendix \ref{Appendix:Proof:LemmaRank}, we can observe that the artificial noise dominates every dimension of the received-signal's subspace. It can be shown that $I(\boldsymbol{\mu}_{\mathcal{S}_{I},\mathcal{S}_{J}}, \boldsymbol{\nu}_{\mathcal{S}_{I},1};Y_{e}) = K(n+1)^{\Gamma}\textrm{log}(P) + o(\textrm{log}(P))$. Likewise, $I(\boldsymbol{\nu}_{\mathcal{S}_{I},1};Y_{e} | \boldsymbol{\mu}_{\mathcal{S}_{I},\mathcal{S}_{J}}) = K(n+1)^{\Gamma}\textrm{log}(P) + o(\textrm{log}(P))$. Therefore, the information leakage is shown to be bounded by $o(\textrm{log}(P))$. It is readily shown that
\begin{align} \nonumber
& ~~~ \lim_{n,P \rightarrow \infty}\Delta_{\mathcal{S_{I}},\mathcal{S_{J}}}^{[e]} = 1 - \frac{I(\mathbf{W}_{\mathcal{K},\mathcal{M}}; Y_{e}^{n})}{nR_{\mathcal{K},\mathcal{M}}} = 1. 
\end{align}
The equivocation at the other receivers can be shown to have limit $1$ following the same method.  
Overall we have the sum SDOF of $d = \frac{K(M-1)n^{\Gamma}}{\mu_{n}} = \frac{K(M-1)n^{\Gamma}}{K(n+1)^{\Gamma} + (M-1)n^{\Gamma}}$, which approaches $\frac{K(M-1)}{K+M-1}$ for large $n$.
\end{proof}

\begin{remark}[Relation to interference channels] \rm
For the case $K=M \geq 3$, the derived SDOF lower bound agrees with the optimal SDOF of the $K$-user interference channel with confidential messages, even when an EE appears in the network. Therefore, treating the $X$ network as the corresponding interference channel yields the best known SDOF lower bound \cite{Jianwei2013}. 
\end{remark}

\begin{remark}[Only one external eavesdropper] \rm
As we see from the SDOF upper bound proof, the $M \times K$ wireless X network with only an external eavesdropper is considered, for which the sum SDOF can be bounded above by $\frac{K(M-1)}{K+M-2+\frac{1}{M}}$. By adopting the same transmission scheme as used in the proof of Lemma \ref{lemma:SDOFXNCMEE}, SDOF of $\frac{K(M-1)}{K+M-1}$ can also be achieved in this network. However, we conjuncture that there might exist another transmission scheme to inject AN only into the subspace of the eavesdropper to obtain higher SDOF. This problem is interesting for our future work.
\end{remark}

\begin{remark}[Comments for perfect secrecy] \rm
In order to study the SDOF, we have shown in the achievability proof that the information leakage at the eavesdropper can be bounded by $o(\textrm{log}(P))$, i.e., $\frac{1}{n}I(\mathbf{W}_{\mathcal{K}-k, \mathcal{M}}; Y_{k}^{n}) = o(\textrm{log}(P))$ at the receiver $k$, which satisfies the secrecy constraint defined in (\ref{Eqn:confidentiality}). In order to achieve perfect secrecy such that $\lim_{n \rightarrow \infty}\frac{1}{n}I(\mathbf{W}_{\mathcal{K}-k, \mathcal{M}}; Y_{k}^{n}) = 0$, we can place a random binning encoder in advance of the ANA operation with a rate penalty $\max_{j \in \{\mathcal{K},e\}-k}\frac{1}{M}I(\boldsymbol{\mu}_{k,m};Y_{j}|\boldsymbol{\mu}_{\mathcal{K},\mathcal{M}}-\boldsymbol{\mu}_{k,m}, \mathbf{H})$ for the message $W_{k,m}$, $k\in \mathcal{K}$ and $m \in \mathcal{M}$, where $\boldsymbol{\mu}_{k,m}$ represents the codeword after binning, $\boldsymbol{\mu}_{\mathcal{K},\mathcal{M}} - \boldsymbol{\mu}_{k,m}$ represents the set formed from $\boldsymbol{\mu}_{\mathcal{K},\mathcal{M}}$ after removing $\boldsymbol{\mu}_{k,m}$, and $\mathbf{H}$ represents the channel state over the alignment block. Note that the argument follows a standard analysis as shown in \cite{Jianwei2013}, \cite{SYang2013} and \cite{Zhao2014a}, because the ANA operation creates an equivalent memoryless channel. It is observed that the rate penalty will not affect the SDOF as shown in the above proof. The detailed proof for achieving perfect secrecy can be found in \cite{ZhaoThesis2015}.
%\begin{align}
%	\nonumber
%	& ~~I(\boldsymbol{\mu}_{k,m};Y_{j}|\boldsymbol{\mu}_{\mathcal{K}-k,\mathcal{M}}, \mathbf{H}) = I(\boldsymbol{\mu}_{k,m};Y_{j}, \boldsymbol{\mu}_{\mathcal{K}-k,\mathcal{M}}| \mathbf{H}) \\ \nonumber
%	& = I(\boldsymbol{\mu}_{k,m};Y_{j}|\mathbf{H}) + I(\boldsymbol{\mu}_{k,m};\boldsymbol{\mu}_{\mathcal{K}-k,\mathcal{M}}| Y_{j}, \mathbf{H}) \\ \nonumber
%	& = I(\boldsymbol{\mu}_{k,m};Y_{j}|\mathbf{H}) + H(\boldsymbol{\mu}_{j,\mathcal{M}}|Y_{j},\mathbf{H}) 
%\end{align}
\end{remark}

\section{Secure Degrees of Freedom of the $M \times K$ XNCM with Reconfigurable Antennas: A Blind Artificial Noise Alignment Approach}
\label{section:SDOFBlind}

In this section, we study the achieved SDOF of the XNCM with reconfigurable antennas, where the each receiver is equipped with one antenna that can switch among $M$ predefined modes. When an antenna switches its mode according to a predefined pattern, it offers a chance to artificially manipulate the channel coherence structure \cite{Tiangao2011}. The received signal for receiver $k$, at time $t$ with antenna mode $\hat{m}_{k}(t)$, is
\begin{align}
	y_{k}(t) = \sum_{m \in \mathcal{M}}h_{km}(\hat{m}_{k}(t))x_{m}(t) + n_{k}(t), ~k \in \mathcal{K},
\end{align}
where $h_{km}(\hat{m}_{k}(t))$ represents the channel coefficient from transmitter $m$ to receiver $k$. 

In the following, we propose a blind ANA that combines blind interference alignment and artificial noise transmission. The transmission follows a similar principle as presented in Section \ref{section:SDOFXNCM}, where we aim to inject artificial noise into the interference space. However, CSIT is not required in this case. IA is based on the channel coherence structure by switching antenna modes. It is worth noting that we assume that the predefined antenna switching mode $\hat{m}_{k}(t)$ is known at the transmitter side; however, it is hidden from all other receivers. Intuitively, the predefined antenna switching functions can be used as the \emph{key} to provide confidentiality. In order to present the achievable scheme in the XNCM, we first introduce the broadcast channel with confidential messages (BCC). We propose a blind ANA scheme in the BCC such that no antenna cooperation is involved, which implies the same achievable SDOF of the corresponding XNCM. 

\begin{definition}\emph{$K$-user $M \times 1$ BCC with reconfigurable antennas.}
\label{definition:BCCM}

Consider the $K$-user broadcast channel, where the transmitter has $M$ antennas and each receiver is equipped with one reconfigurable antenna which can switch among $M$ predefined modes. The received signal at receiver $k$ is
\begin{align} \nonumber
	y_{k}(t) = H_{k}(\hat{m}_{k}(t))X(t) + n_{k}(t), ~ k \in \mathcal{K},
\end{align}
where $H_{k}(\hat{m}_{k}(t)) \!=\! [h_{k1}(\hat{m}_{k}(t))~ h_{k2}(\hat{m}_{k}(t))~ \!\dots\! h_{kM}(\hat{m}_{k}(t))]$ represents the $1 \times M$ channel vector with the mode $\hat{m}(t) \in \mathcal{M}$. $X(t) \in \mathbb{C}^{M \times 1}$ is the transmitted signal and $n_{k}(t) \sim \mathcal{CN}(0,1)$. We assume each channel coefficient is drawn independently from a continuous distribution with finite support. The channel coherence time is assumed to be long enough such that the channels stay constant across a supersymbol, which will be defined in the sequel, and we impose the power constraint $\sum_{t=1}^{n} X^{H}(t)X(t) \leq nP$. We further assume the transmitter sends an independent confidential message $W_{k} \in \mathcal{W}_{k} = [1:2^{nR_{k}}]$ with secrecy rate $R_{k}$, which has to be hidden from the other receivers. With $\mathcal{W} = \{\mathcal{W}_{i}\}_{i=1}^{K}$, a secrecy rate tuple $\mathbf{R} = \{R_{i}\}_{i=1}^{K}$ is achieved if there exists a secret codebook $(n, \mathbf{R}, \mathcal{W})$ to satisfy the following constraints simultaneously: 1) reliability: $\limsup_{n \rightarrow \infty} \textrm{Pr}(\hat{W}_{k} \neq W_{k}) = 0$ for receiver $k$, and 2) confidentiality: $\lim_{n,P \rightarrow \infty}\Delta_{\mathcal{S}}^{[k]} \overset{\triangle}{=} \lim_{n,P\rightarrow \infty}\frac{H(\mathbf{W}_{\mathcal{S}}|Y_{k}^{n},\mathbf{H}^{n})}{H(\mathbf{W}_{\mathcal{S}})} = 1$ where $\mathbf{H}^{n}$ represents the channel state sequence and $\mathbf{W}_{\mathcal{S}} = \{W_{i}: i \in \mathcal{S}\}$ for all $\mathcal{S} \subseteq \mathcal{K}-k$ such that $H(\mathbf{W}_{\mathcal{S}}) > 0$.
\end{definition} 

\begin{lemma}
\label{lemma:SDOFBCCM}
For the $K$-user $M \times 1$ BCC in \emph{Definition} \ref{definition:BCCM}, the SDOF $\frac{K(M-1)}{M+K-1}$ can be achieved.
\end{lemma}

\begin{figure*}[t!]
\centerline{\epsfig{file=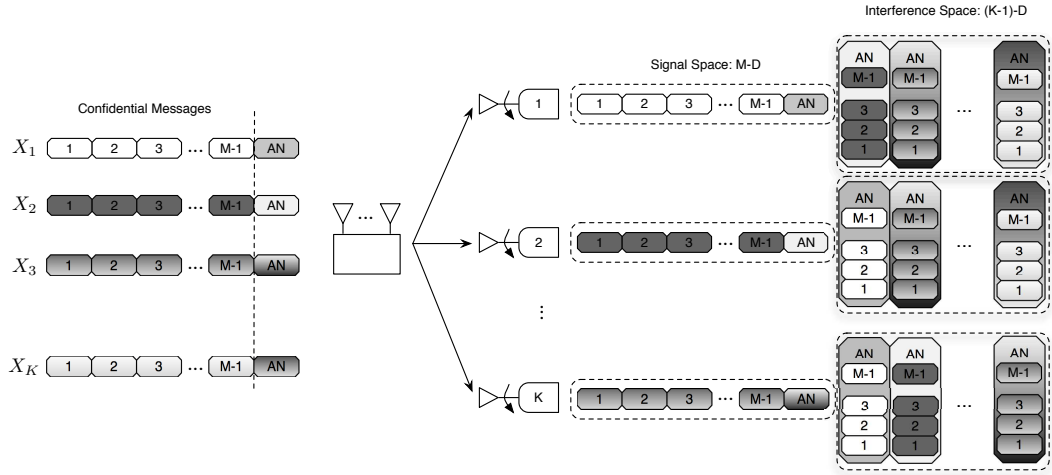,width=140mm}}
\caption{Blind artificial noise alignment for the $K$-user $M \times 1$ BCC.}
\label{Fig:BCCANA}
\end{figure*}

\begin{proof}
The detailed proof can be found in Appendix \ref{Appendix:Proof:TheoremKBCC}. We briefly provide the main idea of achievability as follows: We intend to transmit $M$ independent streams to each receiver, which contain $M-1$ streams of independent confidential messages and $1$ stream of artificial noise. Using interference alignment based on the channel coherence structure due to antenna switching, these $M$ streams cast exactly $M$ dimensions ($M$-D) at the intended receiver, while they overlap in a $1$-D subspace at all the other $K-1$ receivers, which is filled with the artificial noise. Because the artificial noise has the same scaled power compared with the messages in the interference space, the equivocation can be bounded to zero asymptotically. Compared with the achievable DOF of the same network without secrecy constraint, $\frac{M}{M+K-1}$, the price to pay for the confidential message is $\frac{1}{M+K-1}$ for every receiver. As we shall highlight below, in the proposed transmission scheme, the DOF loss $\frac{1}{M+K-1}$ is exactly the dimension occupied by the artificial noise in the transmitted signal. The key idea is illustrated in Fig. \ref{Fig:BCCANA}.
\end{proof}

In order to interpret our main idea, we provide the details of the achievable scheme in the two user $M \times 1$ BCC for $M=2$ and $M=3$. Throughout the transmission scheme, we adopt the antenna switching pattern proposed in \cite{Tiangao2011}, such that the channel coherence structure can be manipulated in a systematic way. 

\subsection{$K=2$, $M=2$}

In this case, each receiver has two modes to switch between. Our goal is to achieve SDOF of $\frac{2}{3}$ in total, which can be obtained by sending one data stream of confidential messages to each receiver in three time slots. To guarantee secrecy, one stream of artificial noise is also sent to protect the confidential messages at the unintended receiver. By antenna switching, we have the following channel coherence structure in three time slots: $\mathbf{H}_{1} = \textrm{diag} \{1,2,1\}$ and $\mathbf{H}_{2} = \textrm{diag} \{1,1,2\}$ \footnote{The abbreviation $\mathbf{H}_{k} = \textrm{diag} \{m_{1},m_{2}, \dots, m_{t}, \dots m_{T}\}$ denotes a block diagonal matrix with the diagonal element $m_{t}$ replaced by $1 \times M$ vectors $H_{k}(m_{t})$, $m_{t} \in \mathcal{M}$. Specifically, $\mathbf{H}_{1} = \textrm{diag} \{1,2,1\} \doteq \textrm{diag} \{H_{1}(1),H_{1}(2),H_{1}(1)\}$ and $\mathbf{H}_{2} = \textrm{diag} \{1,1,2\} \doteq \textrm{diag} \{H_{2}(1),H_{2}(1),H_{2}(2)\}$ where $H_{k}(m)$ is a $1 \times 2$ vector for $k,m \in \{1,2\}$. }.
$\mathbf{H}_{1}$ and $\mathbf{H}_{2}$ are both $3 \times 6$ matrices. These three time slots can be seen as a supersymbol for transmission. We design the $6 \times 2$ beamforming matrix for each receiver constructed using only the $2 \times 2$ identity matrix and zero matrix. In other words, the beamforming matrices do not rely on the value of channel states. Let $\boldsymbol{\Phi}^{[k]}$ denote the beamforming matrix for receiver $k$, $k \in \{1,2\}$. The transmitted signal can be represented as
\begin{align} \nonumber
    X & = \textrm{vec}([X(1)~ X(2)~ X(3)]) \\ \nonumber
    & = \underbrace{\left[\begin{array}{c} \mathbf{I} \\ \mathbf{I} \\ \mathbf{0} \end{array}\right]}_{\boldsymbol{\Phi}^{[1]}} \left[\begin{array}{c} \mu_{1} \\ \nu_{1} \end{array}\right] + \underbrace{\left[\begin{array}{c} \mathbf{I} \\ \mathbf{0} \\ \mathbf{I} \end{array}\right]}_{\boldsymbol{\Phi}^{[2]}} \left[\begin{array}{c} \mu_{2} \\ \nu_{2} \end{array}\right],
\end{align}
where $\mu_{1}$ and $\mu_{2}$ are two independent coded confidential data streams intended for receiver $1$ and $2$, respectively, and $\nu_{1}$ and $\nu_{2}$ are two independent Gaussian artificial noise streams. We assume the powers of the streams scale equally with $P$. The received signal at receiver $1$ is
\begin{align}\nonumber
    Y_{1} & = \left[\begin{array}{c} y_{1}(1) \\ y_{1}(2) \\ y_{1}(3) \end{array}\right] = \mathbf{H}_{1} X + N_{1} = \underbrace{\left[\begin{array}{c} H_{1}(1) \\ H_{1}(2) \\ \mathbf{0} \end{array}\right]}_{\textrm{rank} = 2} \left[\begin{array}{c} \mu_{1} \\ \nu_{1} \end{array}\right] \\ \nonumber 
    & ~~ + \underbrace{\left[\begin{array}{c} H_{1}(1) \\ \mathbf{0} \\ H_{1}(1) \end{array}\right]}_{\textrm{rank} = 1} \left[\begin{array}{c} \mu_{2} \\ \nu_{2} \end{array}\right] + \left[\begin{array}{c} n_{1}(1) \\ n_{1}(2) \\ n_{1}(3) \end{array}\right].
\end{align}
We can see that $\mu_{1}$ and $\nu_{1}$ appear through a rank $2$ matrix, while interference $\mu_{2}$ and $\nu_{2}$ are aligned within a the subspace with rank $1$. It is clear that the signal space and the interference space are orthogonal to each other. Therefore, the confidential data stream $\mu_{1}$ can be resolved from the received signal to obtain $1$ DOF. Moreover, because in the $1$-D interference subspace, $\mu_{2}$ and $\nu_{2}$ are coupled with the same power level, the artificial noise $\nu_{2}$ can protect the confidential data at receiver $1$. Similarly, at receiver $2$, $\mu_{2}$ can achieve $1$ DOF. Therefore, by normalizing the transmission time slots, each receiver achieve DOF of $\frac{1}{3}$. After collecting the whole codeword, we have the equivocation at receiver $1$ as follows, for $\mathcal{S} = \{2\}$,
\begin{align} \nonumber
    \Delta_{\mathcal{S}}^{[1]} 
        & = 1 - \frac{I(W_{2}; Y_{1}^{n}|\mathbf{H}^{n})}{nR_{2}} \\ \nonumber %\\ \label{Eqn:Equivo1_1}
%        \geq 1 - \frac{I(\boldsymbol{\mu}_{2}; Y_{1}^{n}|\mathbf{H}^{n})}{nR_{2}} \\ \nonumber
%        & = 1 - \frac{I(\boldsymbol{\mu}_{2}, \boldsymbol{\nu_{2}}; Y_{1}^{n}|\mathbf{H}^{n}) - I(\boldsymbol{\nu}_{2}; Y_{1}^{n} | \boldsymbol{\mu}_{2},\mathbf{H}^{n})}{nR_{2}} \\ \nonumber % \label{Eqn:Equivo1_2}
        & \geq 1 - \frac{\frac{n}{3}\left(I(\mu_{2}, \nu_{2}; Y_{1}|H_{\mathcal{K}}) - I(\nu_{2}; Y_{1} | \mu_{2}, H_{\mathcal{K}}) \right)}{nR_{2}} \\ \nonumber % \label{Eqn:Equivo1_3}
        & = 1 - \frac{o(\textrm{log}(P))}{\textrm{log}(P) + o(\textrm{log}(P))}
\end{align}
where $\boldsymbol{\mu}_{2}$ and $\boldsymbol{\nu}_{2}$ represent the sequences of $\mu_{2}$ and $\nu_{2}$, respectively, over the codeword length $n$; $H_{\mathcal{K}}$ is the channel matrix for the supersymbol. It is clear that we can show
$\lim_{n,P\rightarrow \infty} \Delta_{\mathcal{S}}^{[1]} = 1$,
to guarantee the confidentiality in the limit of large $n$, and $P$. A similar analysis can be adopted at receiver $2$. Finally, we show that SDOF of $\frac{2}{3}$ is achieved.

\begin{remark} \rm
The main idea of keeping the confidentiality at eavesdroppers is to let the artificial noise fill the interference subspace. When $M > 2$, we will carry out dimension-extension in transmission schemes, which leads to the dimension-expansion on the interference space. Consequently, the key step is to maintain the presence of artificial noise in every dimension of the aligned interference space.
\end{remark}

\subsection{$K=2$, $M=3$} 
We aim to show that SDOF of $\frac{2\times(3-1)}{3+2-1} = 1$ can be achieved. We consider an $8$-slot transmission as a supersymbol, during which four streams of confidential messages can be delivered to each receiver and two streams of artificial noise are used for each receiver. We shall introduce additional unitary beamforming matrices $\mathbf{V}^{[k]}$ in the transmitted signals, with the purpose of maintaining the existence of artificial noise in the interference space. We will show in the sequel that the elements in $\mathbf{V}^{[k]}$ are either $1$ or $0$. Thus they are independent of the value of the channel coefficients. The transmitted signal in a supersymbol for each receiver is designed as follows, respectively:
\begin{align} \nonumber
    X_{1} & = \underbrace{\left[\begin{array}{c c} \mathbf{I} & \mathbf{0} \\ \mathbf{I} & \mathbf{0} \\ \mathbf{I} & \mathbf{0} \\ \mathbf{0} & \mathbf{I} \\ \mathbf{0} & \mathbf{I} \\ \mathbf{0} & \mathbf{I} \\ \mathbf{0} & \mathbf{0} \\ \mathbf{0} & \mathbf{0} \end{array}\right]}_{\boldsymbol{\Phi}^{[1]}} \underbrace{\left[\begin{array}{c c} \mathbf{V}^{[1]}_{\mu} & \mathbf{V}^{[1]}_{\nu} \end{array}\right]}_{\mathbf{V}^{[1]}} \left[\begin{array}{c} \boldsymbol{\mu}_{1} \\ \boldsymbol{\nu}_{1} \end{array} \right], \\ \nonumber 
    X_{2} & = \underbrace{\left[\begin{array}{c c} \mathbf{I} & \mathbf{0} \\ \mathbf{0} & \mathbf{I} \\ \mathbf{0} & \mathbf{0} \\ \mathbf{I} & \mathbf{0} \\ \mathbf{0} & \mathbf{I} \\ \mathbf{0} & \mathbf{0} \\ \mathbf{I} & \mathbf{0} \\ \mathbf{0} & \mathbf{I} \end{array}\right]}_{\boldsymbol{\Phi}^{[2]}} \underbrace{\left[\begin{array}{c c} \mathbf{V}^{[2]}_{\mu} & \mathbf{V}^{[2]}_{\nu} \end{array}\right]}_{\mathbf{V}^{[2]}} \left[\begin{array}{c} \boldsymbol{\mu}_{2} \\ \boldsymbol{\nu}_{2} \end{array} \right],
\end{align}
where $\boldsymbol{\mu}_{k}$ and $\boldsymbol{\nu}_{k}$ are $4 \times 1$ and $2 \times 1$ vectors, respectively, denoting the confidential data streams and artificial noise. The block matrices $\mathbf{V}^{[k]}_{\mu}$ and $\mathbf{V}^{[k]}_{\mu}$ have dimensions $6 \times 4$ and $6 \times 2$, respectively. $\mathbf{I}$ and $\mathbf{0}$ are $3 \times 3$ identity and zero matrices, respectively. By antenna switching, we are able to have the corresponding channel matrices as $\mathbf{H}_{1} = \textrm{diag}\{1,2,3,1,2,3,1,2\}$ and $\mathbf{H}_{2} = \textrm{diag}\{1,1,1,2,2,2,3,3\}$, each representing an $8 \times 24$ matrix. Let the transmitted signal $X = X_{1} + X_{2}$. The received signals over all $8$ slots at receiver $1$ are
\begin{align} \nonumber
	Y_{1} & = \mathbf{H}_{1} X + N_{1} = \mathbf{H}_{1} (X_{1} + X_{2}) + N_{1} \\ \nonumber
          & = \underbrace{\left[\begin{array}{c c} H_{1}(1) & 0 \\ H_{1}(2) & 0 \\ H_{1}(3) & 0 \\ 0 & H_{1}(1) \\ 0 & H_{1}(2) \\ 0 & H_{1}(3) \\ 0 & 0 \\ 0 & 0 \end{array}\right]}_{\textrm{rank} = 6} \left[\begin{array}{c c} \mathbf{V}_{\mu}^{[1]} & \mathbf{V}_{\nu}^{[1]} \end{array}\right] \left[\begin{array}{c} \boldsymbol{\mu}_{1} \\ \boldsymbol{\nu}_{1} \end{array}\right] \\ \nonumber 
          & ~ +\underbrace{\left[\begin{array}{c c} H_{1}(1) & 0 \\ 0 & H_{1}(2) \\ 0 & 0 \\ H_{1}(1) & 0 \\ 0 & H_{1}(2) \\ 0 & 0 \\ H_{1}(1) & 0 \\ 0 & H_{1}(2) \end{array}\right]}_{\textrm{rank} = 2} \left[\begin{array}{c c} \mathbf{V}_{\mu}^{[2]} & \mathbf{V}_{\nu}^{[2]} \end{array}\right] \left[\begin{array}{c} \boldsymbol{\mu}_{2} \\ \boldsymbol{\nu}_{2} \end{array}\right] + N_{1}
\end{align}
where $H_{k}(m)$ and $0$ are $1 \times 3$ vectors. We observe from the signal space that $6$ streams including four data streams in $\boldsymbol{\mu}_{1}$ and two artificial noise streams in $\boldsymbol{\nu}_{1}$ can be resolved almost surely in the $6$-D space. After the alignment, the dimension of the interference space has been reduced to two almost surely. As mentioned, for protecting the confidential messages, we aim to fill the whole interference space with artificial noise, which means the following statement should hold almost surely (a.s.):
\begin{align} \nonumber
    \textrm{rank} \left\{\left[\begin{array}{c c} H_{1}(1) & 0 \\ 0 & H_{1}(2) \end{array}\right] \mathbf{V}_{\nu}^{[2]} \right\} = 2.
\end{align}
One solution for $\mathbf{V}_{\nu}^{[2]}$ can be $\left[\begin{array}{c c c c c c} 1 & 0 & 0 & 0 & 0 & 0 \\ 0 & 0 & 0 & 1 & 0 & 0 \end{array}\right]^{T}$, which yields
\begin{align} \label{Eqn:FullRank}
    \textrm{rank} \left\{\left[\begin{array}{c c} h_{1}(1,1) & 0 \\ 0 & h_{1}(2,1) \end{array} \right]\right\} = 2,
\end{align}
where $h_{1}(1,1)$ and $h_{1}(2,1)$ represent the first elements in the channel vectors $H_{1}(1)$ and $H_{1}(2)$, respectively. It is clear that (\ref{Eqn:FullRank}) holds almost surely. Therefore, $\mathbf{V}^{[2]}$ can be chosen as the elementary matrix,
\begin{align} \label{Eqn:EleMatrix}
    \mathbf{V}^{[2]} = \left[\begin{array}{c c} \mathbf{V}_{\mu}^{[2]} & \mathbf{V}_{\nu}^{[2]} \end{array}\right] = \left[\begin{array}{c c c c} 0 & 0 & 0 & 0 \\ 0 & 1 & 0 & 0 \\ 0 & 0 & 1 & 0 \\ 0 & 0 & 0 & 0 \\ 1 & 0 & 0 & 0 \\ 0 & 0 & 0 & 1\end{array} \vline \begin{array}{c c} 1 & 0 \\ 0 & 0 \\ 0 & 0 \\ 0 & 1 \\ 0 & 0 \\ 0 & 0 \end{array}\right].
\end{align}
By this means, we guarantee the artificial noise can fill the whole interference space. The equivocation $\Delta_{\mathcal{S},1}$ can be shown to have the limit $1$ when $n, P \rightarrow \infty$. A similar argument can be made at receiver $2$, with $\mathbf{V}^{[1]}$ chosen as the same elementary matrix (\ref{Eqn:EleMatrix}) to keep the presence of artificial noise. Therefore, after the alignment, four streams of confidential messages are delivered for each receiver over $8$ slots to achieve SDOF of $1$.

\begin{remark} \rm
We note that the channel coherence structure provides the secret key for receivers to decode. In other words, if the private antenna switching function were known by an eavesdropper, then the eavesdropper could switch its antenna accordingly to mimic the channel coherence structure of the receiver. In this way, the messages could be decoded by the eavesdropper.
\end{remark}

\begin{theorem}
\label{Theorem:SDOFBlind}
For the $M \times K$ XNCM with reconfigurable antennas, the optimal sum SDOF can be bounded as $\frac{K(M-1)}{K+M-1} \leq d \leq \frac{K(M-1)}{K+M-2}$.	
\end{theorem}
\begin{proof}
The converse follows directly from \emph{Theorem} \ref{theorem:SDOFUB}. The achievability follows from the proof of \emph{Lemma} \ref{lemma:SDOFBCCM}, where no antenna cooperation is carried out at transmitters. Therefore, by rearranging the message set, we can obtain the same SDOF. 
\end{proof}

\section{Conclusions}
\label{section:CONC}

In this paper, we have studied the achievable and optimal SDOF of wireless X networks with confidential messages. In particular, we have proposed an SDOF upper bound for the $M \times K$ XNCM. This upper bound has been shown to be tight for $K=2$, with time/frequency varying channels. The achievability of this bound was shown by an ANA scheme, in which artificial noise has been injected into the interference space at receivers. The proposed scheme can be generalized to the $M \times K$ XNCM with time/frequency varying channels for $K \geq 3$, even when an external eavesdropper exists. The achieved SDOF lower bound approaches the upper bound asymptotically with large $M$ and/or $K$. Finally, we have generalized the ANA scheme to the blind case, in which CSIT is not necessarily needed but the channel coherence structure  is known to the transmitters. It is interesting to note that by switching antenna modes artificially, we can not only obtain the intended channel coherence structure for IA but also the secret key for decoding. The achieved SDOF is also asymptotically optimal as the number of nodes in the network approaches infinity. By restricting to linear operations, we have offered a different approach to secrecy coding and interference alignment instead of random binning.

\appendices

\section{Proof of Lemma \ref{lemma:MacBound}}
\label{Appendix:Proof:LemmaMac}
We start from Fano's inequality
\begin{align}
	\nonumber
	& ~~ n \sum_{j\in \mathcal{M}-p} R_{\hat{k},j} = H(\mathbf{W}_{\hat{k},\mathcal{M}-p}) \\ \nonumber
		& \leq I(\mathbf{X}_{\mathcal{M}-p};\mathbf{Y}_{\hat{k}}) + n\epsilon_{1} \\ \nonumber
		& = h(\mathbf{Y}_{\hat{k}}) - h(\mathbf{Y}_{\hat{k}}|\mathbf{X}_{\mathcal{M}-p}) + n \epsilon_{1} \\ \nonumber
		& = h(\mathbf{Y}_{\hat{k}}) - h(\mathbf{H}_{\hat{k}p}X_{p}^{n} + N_{\hat{k}}^{n}|\mathbf{X}_{\mathcal{M}-p}) + n \epsilon_{1} \\ \nonumber
		& = h(\mathbf{Y}_{\hat{k}}) - h(\mathbf{H}_{\hat{k}p}X_{p}^{n} + \mathbf{H}_{\hat{k}p}\tilde{N}^{n} + \tilde{\tilde{N}}^{n})+ n \epsilon_{1} \\ \nonumber
		& \leq h(\mathbf{Y}_{\hat{k}}) - h(\mathbf{H}_{\hat{k}p}X_{p}^{n} + \mathbf{H}_{\hat{k}p}\tilde{N}^{n} + \tilde{\tilde{N}}^{n} | \tilde{\tilde{N}}^{n})+ n \epsilon_{1} \\ \nonumber
		& = h(\mathbf{Y}_{\hat{k}}) - h(X_{p}^{n} + \tilde{N}^{n}) + nO(1),
\end{align}
for some $\epsilon_{1} > 0$ and $\epsilon_{1}$ tends to zero as $n \rightarrow \infty$. $\mathbf{H}_{\hat{k}p} = \textrm{diag}\{h_{\hat{k}p}(1), h_{\hat{k}p}(2), \dots, h_{\hat{k}p}(n)\}$ denotes the channel matrix over $n$ time slots, $X_{p}^{n}$ denotes the $n \times 1$ channel input at transmitter $p$, and $N_{\hat{k}}^{n}$ is the $n \times 1$ noise vector at receiver $\hat{k}$. We note that the third equality follows by defining $N_{\hat{k}}^{n} = \mathbf{H}_{\hat{k}p}\tilde{N}^{n} + \tilde{\tilde{N}}^{n}$ according to the infinite divisibility of the Gaussian distribution, where $\tilde{N}^{n}$ and $\tilde{\tilde{N}}^{n}$ are independent Gaussian processes ($n \times 1$ vectors). It is observed that the variance of the entry $\tilde{n}(t)$ ($t=[1:n]$) in $\tilde{N}^{n}$ is smaller than $\frac{1}{|h_{\hat{k}p}(t)|^{2}}$.

\section{The Converse Proof of Lemma \ref{lemma:SDOFXNCMEE}}
\label{Appendix:Proof:UB_XNCMEE}

In this part we aim to show that the sum SDOF of the $M \times K$ wireless X network with only an external eavesdropper is bounded above by $\frac{K(M-1)}{K+M-2+\frac{1}{M}}$. Therefore, this upperbound also serves as a upper bound for that of the $M \times K$ XNCM-EE. We note that the detailed proof follows a similar approach as the converse proof for the $K$-user interference channel with an external eavesdropper in \cite{Jianwei2012}. Starting from Fano's inequality, we have
\begin{align}
	\nonumber
	& ~~~ n \sum_{i \in \mathcal{K},j\in \mathcal{M}}R_{i,j} = H(\mathbf{W}_{\mathcal{K},\mathcal{M}}) \\ \nonumber
	& \leq I(\mathbf{W}_{\mathcal{K,M}};\mathbf{Y}_{\mathcal{K}}) - I(\mathbf{W}_{\mathcal{K,M}};\mathbf{Y}_{e}) + n\epsilon_{1}\\ \nonumber
	& \leq I(\mathbf{W}_{\mathcal{K,M}};\mathbf{Y}_{\mathcal{K}},\mathbf{Y}_{e}) - I(\mathbf{W}_{\mathcal{K,M}};\mathbf{Y}_{e}) +n\epsilon_{1} \\ \nonumber
	& = I(\mathbf{W}_{\mathcal{K,M}};\mathbf{Y}_{\mathcal{K}}|\mathbf{Y}_{e}) + n\epsilon_{1} \leq I(\mathbf{X}_{\mathcal{M}};\mathbf{Y}_{\mathcal{K}}|\mathbf{Y}_{e}) + n\epsilon_{1} \\ \nonumber
	& = h(\mathbf{Y}_{\mathcal{K}}|\mathbf{Y}_{e}) - h(\mathbf{Y}_{\mathcal{K}}|\mathbf{Y}_{e},\mathbf{X}_{\mathcal{M}}) + n\epsilon_{1} \\ \nonumber
	& = h(\mathbf{Y}_{\mathcal{K}}|\mathbf{Y}_{e}) + no(\textrm{log}(P)) \\ \nonumber
	& = h(\mathbf{Y}_{\mathcal{K}},\mathbf{Y}_{e}) - h(\mathbf{Y}_{e}) + no(\textrm{log}(P)) \\ \nonumber
	& = h(\tilde{\mathbf{X}}_{\mathcal{M}},\mathbf{Y}_{\mathcal{K}},\mathbf{Y}_{e}) - h(\tilde{\mathbf{X}}_{\mathcal{M}}|\mathbf{Y}_{\mathcal{K}},\mathbf{Y}_{e}) - h(\mathbf{Y}_{e}) + no(\textrm{log}(P))\\ \nonumber
	& \leq h(\tilde{\mathbf{X}}_{\mathcal{M}},\mathbf{Y}_{\mathcal{K}},\mathbf{Y}_{e}) \!-\! h(\tilde{\mathbf{X}}_{\mathcal{M}}|\mathbf{Y}_{\mathcal{K}},\mathbf{Y}_{e},\mathbf{X}_{\mathcal{M}}) \\ \nonumber
	& ~~~~~  \!-\! h(\mathbf{Y}_{e}) \!+\! no(\textrm{log}(P)) \\ \nonumber
	& = h(\tilde{\mathbf{X}}_{\mathcal{M}}) + h(\mathbf{Y}_{\mathcal{K}},\mathbf{Y}_{e}|\tilde{\mathbf{X}}_{\mathcal{M}}) - h(\mathbf{Y}_{e}) + no(\textrm{log}(P)) \\ 
	\label{Eqn:NoiseEntropy}
	& \leq h(\tilde{\mathbf{X}}_{\mathcal{M}}) - h(\mathbf{Y}_{e}) + no(\textrm{log}(P)) \\
	\label{Eqn:DiminishInput}
	& \leq h(\tilde{\mathbf{X}}_{\mathcal{M}-1}) + no(\textrm{log}(P)) \\ \nonumber
	& \leq \sum_{p \in \mathcal{M}-1} h(\tilde{\mathbf{X}}_{p}) + no(\textrm{log}(P)) \\ \nonumber
	& \leq \sum_{p \in \mathcal{M}-1} \left(h(\mathbf{Y}_{\hat{k}}) - n\sum_{j \in \mathcal{M}-p} R_{\hat{k},j} \right) + no(\textrm{log}(P)),
\end{align}
for $\epsilon_{1}= o(\textrm{log}(P))$, where $\tilde{\mathbf{X}}_{\mathcal{M}} = \{X_{i}^{n} + \tilde{N}_{i}^{n}\}_{i \in \mathcal{M}}$, and the entry $\tilde{n}_{i}(t)$ of $\tilde{N}_{i}^{n}$ has variance smaller than $\frac{1}{|h_{\hat{k}i}(t)|^{2}}$ ($t=[1:n]$). (\ref{Eqn:NoiseEntropy}) follows from the fact that by conditioning on the noisy input, only effective noise contributes at the outputs which have differential entropy bounded by $no(\textrm{log}(P))$. (\ref{Eqn:DiminishInput}) follows from the following inequalities:
\begin{align}
	\nonumber
    h(\mathbf{Y}_{e}) & = h(\sum_{i=1}^{M}\mathbf{H}_{ei}X_{i}^{n}+N_{e}^{n}) \\ \nonumber
	& \geq h(\mathbf{H}_{e1}X_{1}^{n} + N_{e}^{n}) = h(\tilde{X}^{n}_{1}) + no(\textrm{log}(P)),
\end{align}
where $\mathbf{H}_{ei}$ is the $n \times n$ channel matrix and $X_{i}^{n}$ is the $n \times 1$ transmitted signal vector. 

Manipulating both sides, we have
\begin{align}
	\nonumber
	& (K+M-1)H(\mathbf{W}_{\mathcal{K},\mathcal{M}}) - H(\mathbf{W}_{\mathcal{K},\mathcal{M}-1}) \\ \nonumber 
	& ~~~~~~\leq (M-1)\sum_{\hat{k} \in \mathcal{K}} h(\mathbf{Y}_{\hat{k}}). 
\end{align}
Henceforth, we have the sum SDOF $d \leq \frac{K(M-1)}{K+M-2+ \frac{1}{M}} < \frac{K(M-1)}{K+M-2}$.

\section{A Preliminary Lemma from Linear Algebra}
\label{Appendix:Proof:LemmaRank}

\begin{lemma} \label{Theorem:lemma:RankInc}
For matrices $\mathbf{A} \in \mathbb{C}^{M \times N}$, $\mathbf{B}_{i} \in \mathbb{C}^{M \times T_{i}}$, with $M \geq N \geq T_{i}$ and $i \in \{1,2,\dots,I\}$, if $\textrm{span}(\mathbf{B}_{i}) \subseteq \textrm{span}(\mathbf{A})$, then we have
\begin{align} \nonumber
    \textrm{rank}(\mathbf{A}\mathbf{A}^{H} + \sum_{i=1}^{I} \lambda_{i} \mathbf{B}_{i}\mathbf{B}_{i}^{H}) = \textrm{rank}(\mathbf{A}\mathbf{A}^{H}),~~ \textrm{for}~\lambda_{i} \geq 0.
\end{align}
\end{lemma}

\begin{proof}
Let us assume that $\mathbf{A}$ has rank $d$, where $d \leq N$. Apply the compact singular value decomposition (SVD) on $\mathbf{A}$ and $\mathbf{B}_{i}$. Let $\mathbf{U} = [U_{1}, U_{2}, \dots, U_{d}] \in \mathbb{C}^{M \times d}$ and $\mathbf{\Gamma}_{i} = [\Gamma_{1}, \Gamma_{2}, \dots, \Gamma_{d_{i}}] \in \mathbf{C}^{M \times d_{i}}$ be the left singular vector matrices of $\mathbf{A}$ and $\mathbf{B}_{i}$, respectively, corresponding to the nonzero singular values. We have $d_{i} \leq d$ because of $\textrm{span}(\mathbf{B}_{i}) \subseteq \textrm{span}(\mathbf{A})$. Our goal is to show that the matrix $\mathbf{A}\mathbf{A}^{H} + \sum_{i=1}^{I} \lambda_{i}\mathbf{B}_{i}\mathbf{B}_{i}^{H}$ has at most $d$ linearly independent column vectors. Because $U_{j}$ are the eigenvectors of matrix $\mathbf{A}\mathbf{A}^{H}$, any column vector of $\mathbf{A}\mathbf{A}^{H}$ can be written as the linear combination of $U_{j}$, for $j \in [1:d]$. We also note that the $U_{j}$ form an basis for the column space of $\mathbf{A}$. We then show that any column vector of $\mathbf{B}_{i}\mathbf{B}_{i}^{H}$ can also be written as a linear combination of the $U_{j}$. Choose an arbitrary column vector in $\mathbf{B}_{i}\mathbf{B}_{i}^{H}$, denoted as $Q$. We have $Q = \sum_{k=1}^{d_{i}} \alpha_{k} \Gamma_{k}$. Because $\textrm{span}(\mathbf{B}_{i}) \subseteq \textrm{span}(\mathbf{A})$, $\Gamma_{k} = \sum_{j=1}^{d} \beta_{j} U_{j}$ which shows that any left singular vector of $\mathbf{B}_{i}$ can be written as a linear combination of $U_{j}$ (the basis of $\textrm{span}(\mathbf{A})$). Thus, $Q = \sum_{k=1}^{d_{i}} \alpha_{k} (\sum_{j=1}^{d} \beta_{j} U_{j})$. This shows that any column vector of $\mathbf{A}\mathbf{A}^{H} + \sum_{i=1}^{I} \lambda_{i}\mathbf{B}_{i}\mathbf{B}_{i}^{H}$ can be written as a linear combination of $d$ independent vectors, and thus the rank of of $\mathbf{A}\mathbf{A}^{H} + \sum_{i=1}^{I} \mathbf{B}_{i}\mathbf{B}_{i}^{H}$ is at most $d$.

On the other hand, because $\mathbf{A}\mathbf{A}^{H}$ and $\lambda_{i} \mathbf{B}_{i}\mathbf{B}_{i}^{H}$ are \emph{Hermitian} matrices, it is clear
\begin{align} \nonumber
	\textrm{rank}(\mathbf{A}\mathbf{A}^{H} + \sum_{i=1}^{I} \lambda_{i}\mathbf{B}_{i}\mathbf{B}_{i}^{H}) \geq \textrm{rank}(\mathbf{A}\mathbf{A}^{H}) = d.
\end{align}
Therefore, we conclude the rank of matrix $\mathbf{A}\mathbf{A}^{H} + \sum_{i=1}^{I} \lambda_{i}\mathbf{B}_{i}\mathbf{B}_{i}^{H}$ is $d$ to complete the proof.
\end{proof}

Lemma \ref{Theorem:lemma:RankInc} provides the fundamentals to analyze the dimension of the aligned subspace in our achievable schemes. For instance, if the interference space $\textrm{span}(\mathbf{B}_{i})$ is aligned into the subspace $\textrm{span}(\mathbf{A})$, then by Lemma \ref{Theorem:lemma:RankInc}, we know the aligned subspace is dominated by $\mathbf{A}$.

\section{Proof of Lemma \ref{lemma:SDOFBCCM}}
\label{Appendix:Proof:TheoremKBCC}

For the $K$-user BCC with $M$ reconfigurable modes for the antenna at each receiver, we aim to show that SDOF of $\frac{K(M-1)}{M+K-1}$ can be achieved almost surely. Let $\alpha = (M-1)^{K-1}$ and $\beta=(K+M-1)\alpha$. Consider an $\alpha$ symbol extension over the original channel, and construct the supersymbol in $\beta$ time slots. Then, it is equivalent to show that $(M-1)\alpha$ confidential streams can be delivered to each receiver in $\beta$ time slots, which yields $d = \frac{(M-1)\alpha}{\beta} = \frac{(M-1)^{K}}{(K+M-1)(M-1)^{K-1}} = \frac{M-1}{M+K-1}$ achieved for each receiver. The transmitted signal for the receiver $k \in \mathcal{K}$ can be designed as
\begin{align} \nonumber
	X_{k} & = \textrm{vec}[X_{k}(1)~ X_{k}(2)~ \dots X_{k}(\beta)] \\ \nonumber 
	& = \boldsymbol{\Phi}^{[k]} \underbrace{\left[\mathbf{V}_{\boldsymbol{\mu}}^{[k]}~ \mathbf{V}_{\boldsymbol{\nu}}^{[k]}\right]}_{\mathbf{V}^{[k]}} \left[\begin{array}{c} \boldsymbol{\mu}_{k} \\ \boldsymbol{\nu}_{k} \end{array}\right],
\end{align}
where $\boldsymbol{\mu}_{k}$ is the $(M-1)\alpha \times 1$ confidential symbol vector intended for receiver $k$, and $\boldsymbol{\nu}_{k}$ is the $\alpha \times 1$ artificial noise stream. Thus the total dimension of the signal vector is $M\alpha$. $\boldsymbol{\Phi}^{[k]}$ is the $M\beta \times M\alpha$ beamforming matrix, $\mathbf{V}^{[k]}$ is a unitary matrix with dimensions $M\alpha \times M\alpha$, in which the block matrix $\mathbf{V}_{\boldsymbol{\nu}}^{[k]} \in \mathbb{C}^{M\alpha \times \alpha}$ is designed to make sure that the artificial noise can fill the whole interference space almost surely. We adopt the same antenna switching mode as proposed in \cite{Tiangao2011}, which gives the specific pattern for each $\mathbf{H}_{k}$ with dimensions $\beta \times M\beta$. For simplicity, we omit the details of $\mathbf{H}_{k}$ and $\boldsymbol{\Phi}^{[k]}$ (please refer to \cite{Tiangao2011}), and instead the emphasis is placed on the submatrix $\mathbf{V}_{\boldsymbol{\nu}}^{[k]}$ for transmitting artificial noise.

Consider the received signal at receiver $k$, which is
\begin{align} \nonumber
	Y_{k} & = \mathbf{H}_{k}\sum_{j \in \mathcal{K}}X_{k} + N_{k} = \mathbf{H}_{k}\boldsymbol{\Phi}^{[k]}\mathbf{V}^{[k]}\left[\begin{array}{c} \boldsymbol{\mu}_{k} \\ \boldsymbol{\nu}_{k} \end{array}\right] \\ \nonumber 
	& ~~ + \sum_{j \in \mathcal{K}-k} \mathbf{H}_{k} \boldsymbol{\Phi}^{[j]} \mathbf{V}^{[j]} \left[\begin{array}{c} \boldsymbol{\mu}_{j} \\ \boldsymbol{\nu}_{j} \end{array}\right] + N_{k}.
\end{align}
We first set all $\mathbf{V}^{[k]}$ to be the same choice, such that $\mathbf{V}^{[k]} = \mathbf{V}$ for all $k \in \mathcal{K}$. Let $\mathbf{G}_{k} = \mathbf{H}_{k} \boldsymbol{\Phi}^{[k]} \mathbf{V}$, and $\mathbf{Q}_{j} = \mathbf{H}_{k} \boldsymbol{\Phi}^{[j]} \mathbf{V}$ for $j \in \mathcal{K}-k$. We have
\begin{align}
	Y_{k} = \mathbf{G}_{k} \left[\begin{array}{c} \boldsymbol{\mu}_{k} \\ \boldsymbol{\nu}_{k} \end{array}\right] + \sum_{j \in \mathcal{K}-k} \mathbf{Q}_{j} \left[\begin{array}{c} \boldsymbol{\mu}_{j} \\ \boldsymbol{\nu}_{j} \end{array}\right] + N_{k}.
\end{align}
By choosing $\mathbf{H}_{k}$ and $\boldsymbol{\Phi}^{[k]}$ ($k \in \mathcal{K}$) as in \cite{Tiangao2011}, it is shown that $\mathbf{H}_{k}\boldsymbol{\Phi}^{[k]}$ and $\mathbf{H}_{k}\boldsymbol{\Phi}^{[j]}$ ($j \in \mathcal{K}-k$) are all orthogonal. Because the isometry of the unitary matrix $\mathbf{V}$, which preserves the orthogonality and matrix rank, it is clear that $\mathbf{G}_{k}$ and $\mathbf{Q}_{j}$ ($j \in \mathcal{K}-k$) are all orthogonal. Moreover, $\mathbf{G}_{k}$ is shown to have rank $M\alpha$, which implies that $\boldsymbol{\mu}_{k}$ and $\boldsymbol{\nu}_{k}$ can be resolved almost surely. Therefore, $\boldsymbol{\mu}_{k}$ has rate $R_{k} = \frac{(M-1)\alpha}{\beta} \textrm{log}(P) + o(\textrm{log}(P))$. For the interference subspace, $\mathbf{Q}_{j}$ is shown to have rank $\alpha$. To guarantee that the artificial noise can fill the interference space, it suffices to show that
\begin{align} \nonumber
	\textrm{rank}\{\mathbf{H}_{k}\boldsymbol{\Phi}^{[j]}\mathbf{V}_{\boldsymbol{\nu}}\} = \alpha, ~ \textrm{a.s.}
\end{align}
which yields
\begin{align}
	\textrm{rank} \left\{\left[\begin{array}{c c c c} H'_{j}(1) & ~ & ~ & ~ \\ ~ & H'_{j}(2) & ~ & ~ \\ ~ & ~ & \ddots & ~ \\ ~ & ~ & ~ & H'_{j}(\alpha) \end{array}\right] \mathbf{V}_{\boldsymbol{\nu}}\right\} = \alpha, ~ \textrm{a.s.}
\end{align}
where for $t=[1:\alpha]$
\begin{align}
	H'_{j}(t) = \left\{ \begin{array}{l l} H_{j}(M) & \quad \textrm{if}~ t~ \textrm{mod}~ M = 0 \\ H_{j}(t~ \textrm{mod}~ M) & \quad \textrm{otherwise}. \end{array} \right.
\end{align}
We choose $\mathbf{V}_{\boldsymbol{\nu}} = \mathbf{I} \otimes \Theta$, with $\Theta = [1~0~\dots~0]^{T}$ denoting the $M \times 1$ elementary vector, and $\mathbf{I}$ denoting the identity matrix with dimension $\alpha \times \alpha$. We have
\begin{align}
	\textrm{rank}\left\{\textrm{diag}\{h'_{j}(1,1), h'_{j}(2,1), \dots, h'_{j}(\alpha,1)\} \right\} = \alpha, ~ \textrm{a.s.}
\end{align}
where $h'_{j}(i,1)$ represents the first element of the channel vector $H'_{j}(i)$. It is clear that the above statement holds. Thus, $\mathbf{V}$ can be chosen as an elementary matrix with the block $\mathbf{V}_{\boldsymbol{\nu}}^{[k]} = \mathbf{I} \otimes \Theta$. Let $\bar{\mathbf{Q}}_{j} = \mathbf{H}_{k} \boldsymbol{\Phi}^{[j]} \mathbf{V}_{\boldsymbol{\nu}}$ and $\mathcal{S} = \mathcal{K}-k$. We consider the information leakage at receiver $k$ after the whole codeword length $n$,
\begin{align} \nonumber
	& \frac{1}{n}I(\mathbf{W}_{\mathcal{S}}; Y_{k}^{n}, H^{n}) \leq \frac{1}{\beta} I(\boldsymbol{\mu}_{\mathcal{S}};Y_{k}|H_{\mathcal{K}}) \\ \nonumber 
	& ~~~~~= \frac{1}{\beta}\left(I(\boldsymbol{\nu}_{\mathcal{S}},\boldsymbol{\mu}_{\mathcal{S}};Y_{k}|H_{\mathcal{K}}) - I(\boldsymbol{\nu}_{\mathcal{S}};Y_{k}|\boldsymbol{\mu}_{\mathcal{S}},H_{\mathcal{K}})\right) \\ \nonumber
	& ~~~~~\leq \frac{1}{\beta}\textrm{log} \det\left( P \sum_{j \in \mathcal{S}} \mathbf{Q}_{j}\mathbf{Q}_{j}^{H} + I \right) \\ \nonumber 
	& ~~~~~~~ - \frac{1}{\beta} \textrm{log} \det\left( P \sum_{j \in \mathcal{S}} \bar{\mathbf{Q}}_{j}\bar{\mathbf{Q}}_{j}^{H} + I \right) + o(\textrm{log}(P)) \\ \nonumber
		& ~~~~~= o(\textrm{log}(P)),
\end{align}
where the last equality follows from the fact that $\mathbf{Q}_{j}\mathbf{Q}_{j}^{H}$ and $\bar{\mathbf{Q}}_{j}\bar{\mathbf{Q}}_{j}^{H}$ have the same rank almost surely. Then,
\begin{align}
	\Delta_{\mathcal{S}}^{[k]} & = 1 - \frac{I(\mathbf{W}_{\mathcal{S}}; Y_{k}^{n}|H_{\mathcal{K}})}{nR_{\mathcal{S}}} \geq 1 - \frac{o(\textrm{log}(P))}{d_{\mathcal{S}}\textrm{log}(P) + o(\textrm{log}(P))}
\end{align}
with $d_{\mathcal{S}} = \frac{(K-1)(M-1)\alpha}{\beta}$, which yields $\lim_{n,P \rightarrow \infty} \Delta_{\mathcal{S}}^{[k]} = 1$. This concludes the proof.

\bibliographystyle{IEEEtran}
\bibliography{IEEEabrv,MyLibraryANA}

\end{document}